\documentclass[12pt,draftcls,onecolumn,twoside]{IEEEtran} 
\usepackage[latin1]{inputenc}
\usepackage{amsmath,amssymb,color,bbm,enumerate,ntheorem}
\usepackage{amsfonts}
\usepackage{graphicx}
\usepackage{xspace}
\usepackage{hyperref}

\newtheorem{lemma}{Lemma}

\newtheorem{theorem}{Theorem}
\newtheorem{definition}{Definition}

\newcommand{\bs}{\boldsymbol}

\newcommand{\A}{\mathcal{A}}
\newcommand{\C}{\mathbb{C}}

\newcommand{\E}{\mathbb{E}}

\renewcommand{\L}{\mathcal{L}}
\newcommand{\Leb}{\mathrm{Leb}}
\newcommand{\N}{\mathcal{N}}
\newcommand{\bP}{\mathbb{P}}
\newcommand{\PF}{\mathrm{PF}}
\newcommand{\R}{\mathbb{R}}

\newcommand{\V}{\mathcal{V}}
\newcommand{\Tr}{\mathop{\mathrm{Tr}}}
\newcommand{\Var}{\mathrm{\mathop{Var}}}
\newcommand{\W}{\mathcal{W}}
\newcommand{\w}{\boldsymbol{w}}
\newcommand{\x}{\boldsymbol{x}}
\newcommand{\y}{\boldsymbol{y}}

\newcommand{\diag}{\textrm{diag}}
\def\rmd{\mathrm{d}}

\def\rmi{\mathrm{i}}
\def\p{\alpha}\def\1{\mathbbm{1}}

\def\cp{\mathop{\stackrel{\mathbbm{P}}{\longrightarrow}}}

\begin{document}

\title{Neyman-Pearson Detection of a\\ Gaussian Source 
using Dumb Wireless Sensors}

\author{Pascal
  Bianchi,~\IEEEmembership{Member,~IEEE,} J\'er\'emie Jakubowicz,~\IEEEmembership{Member,~IEEE,} \\ and Fran\c cois Roueff
  \thanks{The authors are with Institut Telecom / Telecom ParisTech /
    CNRS LTCI, France.} 
  \thanks{e-mails:
    \texttt{\{bianchi,jakubowi,roueff\}@telecom-paristech.fr}}
}
\date{}

\markboth{submitted to \emph{IEEE Transactions on Signal Processing}}
{Bianchi\emph{et.al.}}

\maketitle

\begin{abstract} 
  We investigate the performance of the Neyman-Pearson detection of a
  stationary Gaussian process in noise, using a large wireless sensor network
  (WSN).  In our model, each sensor compresses its observation sequence using a
  linear precoder. The final decision is taken by a fusion center (FC) based on
  the compressed information.  Two families of precoders are studied: random
  iid precoders and orthogonal precoders.  We analyse their performance in
  the regime where both the number of sensors $k$ and the number of samples $n$
  per sensor tend to infinity at the same rate, that is, $k/n\to c \in (0,1)$.
  Contributions are as follows. {\bf 1)} Using results of random matrix theory
  and on large Toeplitz matrices, it is proved that the miss probability of the
  Neyman-Pearson detector converges exponentially to zero, when the above
  families of precoders are used. Closed form expressions of the corresponding
  error exponents are provided.  {\bf 2)} In particular, we propose a practical
  orthogonal precoding strategy, the Principal Frequencies Strategy (PFS),
  which achieves the best error exponent among all orthogonal strategies, and
  which requires very few signaling overhead between the central processor and
  the nodes of the network.  {\bf 3)} Moreover, when the PFS is used, a
  simplified low-complexity testing procedure can be implemented at the FC. We
  show that the proposed suboptimal test enjoys the same error exponent as the
  Neyman-Pearson test, which indicates a similar asymptotic behaviour of the
  performance. We illustrate our findings by numerical experiments on some
  examples. 
\end{abstract}
	
\section{Introduction}

The design of powerful tests allowing to detect the presence of a
stochastic signal using large WSN's is a crucial issue in a wide range
of applications. We investigate the Neyman-Pearson detection of a
Gaussian signal using a wireless network of $k$ sensors. Each sensor
observes a finite sample of the signal of interest, corrupted by
additive noise, and then forwards some information towards the FC
which takes the final decision. Neyman-Pearson detection of Gaussian
signals using large sensor networks has been thoroughly investigated
in the literature (see for instance \cite{STP06,hacIT09} and
references therein). In such works, the FC is assumed to have a perfect
knowledge of the observation sequence of each sensor. Unfortunately,
in a WSN, the amount of information forwarded by each sensor node to
the FC is usually limited, due to channel capacity constraints. Thus,
in practice, each sensor node must compress its information in some
way before transmission to the FC. This compression step of course
degrades the performance of the detection. A large number of works has
been devoted to the determination of relevant compression strategies,
essentially within the framework of distributed detection
\cite{blu97,xia06}.  In these works, the data is locally processed by
each sensor: Typically, a local Neyman-Pearson test is made by each
node, based on the knowledge of the probabilistic law of the source to
be detected.  Unfortunately, such approaches require at the same time
that each sensor possesses a significant computational ability
allowing involved processing of its data, and that each sensor has a
full knowledge of the source statistics.  On the opposite, this paper
investigates the case of {\sl dumb} WSN. By this term, we refer to the
case where:
\begin{itemize}
\item Individual sensor nodes are not aware of their mission and their
  environment.  They process the observed data with no or few instructions from
  the central processor.
\item The processing abilities of each sensor node are limited due either to
  hardware or energy constraints.
\end{itemize}
Dumb WSN are of practical interest because they are simple, flexible
(\emph{i.e.}, easily reconfigurable as a function of the sensor network's
mission) and avoid an excess of signaling overhead in the network.

The aim of this paper is to propose and to study different compressing
strategies which satisfy the above constraints 
and which are attractive in terms of detection performance.  The paper
is organized as follows.  Section~\ref{sec:model} introduces the
signal model. Each sensor is assumed to observe $n$ noisy samples of a
stationary (correlated) Gaussian source.  The spectral density $f$ of
the source is known at the FC but is unknown at the sensor nodes.  The
aim is to detect the presence of the source.  To that end, each node
forwards a compressed version of its observed sequence to the FC.  In
our model, the latter compression is achieved through simple (linear)
processing of the data, allowing this way for low cost implementation.
We refer to this step as \emph{linear precoding}.
Section~\ref{sec:LRT} introduces the problem of the detection of the
presence of the source (hypothesis $H_1$) versus the hypothesis that
only thermal noise is observed (hypothesis $H_0$). It is well known
that a uniformly most powerful (UMP) test is obtained by the celebrated
Neyman-Pearson procedure.  The corresponding test is derived in
Subsection~\ref{sec:exprLRT}.  Intuitively, the good detection
performance of the Neyman-Pearson test fundamentally relies on the relevant
selection of the linear precoders used at the sensor nodes.  Useful
families of linear precoders are introduced, namely random iid
precoders and orthogonal precoders.  The detection performance
associated with each of these families is studied in the 
asymptotic regime where both the number $k$ of sensors and the number
$n$ of observations per sensor tend to infinity at the same rate
($k,n\to\infty$, $k/n\to c$ where $c\in (0,1)$).  More precisely, we
show in Section~\ref{sec:errexp} that for any fixed $\alpha\in
(0,1)$, the miss probability of the NP test of level $\alpha$
converges exponentially to zero.  Error exponents are characterized
and compared for the precoding strategies of interest.  In particular, it is
proved that the so-called Principal Frequencies Strategy (PFS)
achieves the best error
exponent among all orthogonal strategies. Numerical computations of all the obtained error
exponents on some examples conclude this section.
In the case where PFS is
used, a suboptimal (non UMP) test is proposed in Section~\ref{sec:approxPFS}.  Based on the proof of a Large Deviation
Principle governing the proposed test statistics, it is shown that our
suboptimal test achieves the same error exponent as the Neyman-Pearson
test. Finally, Section~\ref{sec:simus} is devoted to the simulations.

\subsection*{Notations}

Column vectors are represented by bold symbols.  Notation $\|\bs y\|$ denotes
the Euclidean norm of vector $\bs y$.  
We denote by $\Leb$ the Lebesgue measure restricted to $[-\pi,\pi]$. For any function
$f:[-\pi,\pi]\to\R$, we use notation
$f^{-1}(A)=\{\omega\in[-\pi,\pi]\,:\,f(\omega)\in A\}$ for the inverse image of $A$, and we
denote by $\Leb \circ f^{-1}$ the image measure of $\Leb$ by $f$, \emph{i.e.}
which composes $\Leb$ with $f^{-1}$. 
For any square matrix $M$, $\rho(M)$ denotes its spectral radius.
Finally, $I_k$ denotes the $k× k$ identity matrix.
\section{The Framework}
\label{sec:model}

\subsection{Observation model at the sensor nodes}

Consider a set of $k$ sensors whose aim is to detect the presence of a certain
source signal $x(0), x(1), x(2)\dots$. Each sensor $i=1\dots k$ collects $n$
noisy samples of the source signal. We assume that $n\geq k$. Denote by $\y_i =
[y_i(0),\dots, y_i(n-1)]^T$ the $n × 1$ data vector observed by sensor $i$. For
each $i=1,\dots,k$, we consider the following signal model:
\begin{equation}
  \label{eq:model}
  \y_i = \x + \w_i,
\end{equation}
where $\x=[x(0),\dots, x(n-1)]^T$ contains the time samples extracted from a
zero mean stationary Gaussian process $x$ with known spectral density function
$f(\omega)$, $\omega\in [-\pi,\pi)$.  Vector $\w_i=[w_i(0),\dots, w_i(n-1)]^T$
is a zero mean white Gaussian process which stands for the thermal noise of
sensor~$i$. We denote by $\sigma^2$ the variance of $w_i(0)$ which is assumed to
be the same for all $i$. Random vectors $\x$, $\w_1,\dots,\w_k$ are supposed to
be independent. 
In the usual framework of Gaussian source detection, the aim is to detect
whether the signal $x$ of interest is present.  Formally, this reduces to the
following hypothesis testing problem:
\begin{eqnarray*}
 	H_1: & & \y_i = \x  + \w_i,\;\;\forall i =1\dots k\\
 	H_0: & & \y_i = \w_i,\;\; \forall i =1\dots k\ .
\end{eqnarray*}
In this paper, we make the following technical assumptions on
the spectral density $f$:
\begin{description}
\item[{\bf A1.}] {\sl The spectral density $f$ is continuous on $[-\pi,\pi]$.}


\item[{\bf A2.}] {\sl Measure $\Leb \circ f^{-1}$ does not put mass on points.}
\end{description}
Assumption {\bf A2} says that $f$ cannot be constant over a set of positive Lebesgue
measure (say, an interval of positive length). This e.g. rules out a white noise
for $x$. On the other hand any ARMA process $x$ that is not a white noise satisfies
Assumptions {\bf A1} and {\bf A2}.

\subsection{Assumptions and constraints on the network}

We assume that the decision is taken by a distant node (the fusion
center).  The latter is supposed to have a perfect knowledge of the noise
variance $\sigma^2$ and of the spectral density $f$ of the signal $x$ to be
detected. In this paper, we are interested in WSN satisfying the following
constraints.  \smallskip

\subsubsection{Communication constraint}
In an ideal WSN architecture, each sensor $i=1,\dots,k$ would transmit all
available observations $y_i(0),\dots,y_i(n-1)$ to the FC.  Unfortunately,
perfect forwarding of the whole information sequence $\bs y_i$ by each sensor
$i$ is impractical in a large number of situations, the amount of information
transmitted by each sensor node to the fusion center being usually limited.  In
this paper, we consider the case where only a compressed version of $\y_i$ is
likely to be forwarded.  More precisely, we assume that each sensor $i$ forwards
a single scalar $z_i$ to the fusion center, where $z_i$ is a certain mapping of
the sequence $\bs y_i$ received by sensor $i$.

\subsubsection{Signaling overhead constraint}
Depending on the particular mission of the network or on the particular spectral
density $f$ to be detected, the network should be easily reconfigurable using a
limited number of feedback bits from the fusion center to the sensors. In the
sequel we assume that the spectral density $f$ is known at the fusion center but
is unknown (or at most partially known) at the sensor nodes. \smallskip

\subsubsection{Complexity constraint}
Only low complexity data processing is likely to be implemented at the sensors'
side. More precisely, we assume that each sensor node $i=1\dots k$ forwards a
linear combination
\begin{equation}
\label{eq:zi}
z_i=\bs a_i^T \y_i 
\end{equation}
of its observation sequence $\y_i$ to the fusion center, where $\bs a_i$ is a
$n× 1$ vector to be determined.  
Figure~\ref{fig:synopt} provides an illustration of the
sensing scheme.
\begin{figure}[t]
  \centering
  \includegraphics[width=0.8\linewidth]{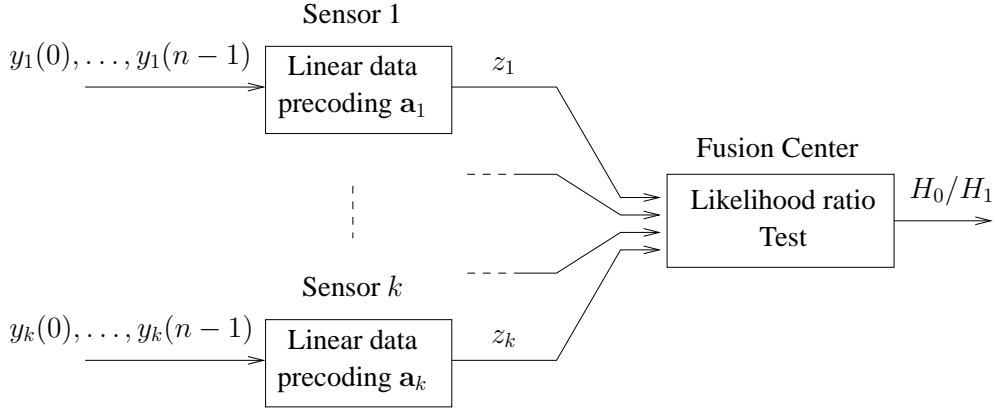}
  \caption{Sensor network using linear precoding at the nodes.}
  \label{fig:synopt}
\end{figure}
Such a set of vectors $\bs a_1,\dots,\bs
a_k$ will be refered to as a \emph{linear precoder}.  The $n× k$ matrix
${A}_n=[\bs a_1,\dots, \bs a_k]$ will be refered to as the \emph{precoding
  matrix}.


\section{Likelihood Ratio Test}
\label{sec:LRT}

\subsection{Expression of the Likelihood Ratio}
\label{sec:exprLRT}

We denote by $\bP_0$ and $\bP_1$ the probability under $H_0$ and $H_1$ and by
$\E_0$ and $\E_1$ the corresponding expectations. Denote by $\bs z =
[z_1,\dots,z_k]^T$ the available $k× 1$ observation vector at the fusion center,
where for each $i$, $z_i$ is defined by~(\ref{eq:zi}). We denote by
$p_{0}:\R^k\to\R_+$ and $p_{1}:\R^k\to\R_+$ the joint probability density
function of $z_1,\dots,z_k$ under hypotheses $H_0$ and $H_1$ respectively.  Due
to the celebrated Neyman-Pearson's Lemma, the Likelihood Ratio Test (LRT) is
uniformly most powerful.  The LRT rejects the null hypothesis for large values
of the log-likelihood ratio (LLR) defined by:
\begin{equation}
\label{eq:npfc}
  \L_{A_n} = \log \frac{p_{1}(\bs z)}{p_{0}(\bs z)} \ .
\end{equation}
In the above definition, the lowerscript $A_n$ has been introduced to recall
that the distribution of the random variable $\L_{A_n}$ depends on the
particular choice of the precoding matrix ${A}_n=[\bs a_1,\dots, \bs a_k]$.  We
now derive a closed form expression of the LLR $\L_{A_n}$.  It is worth noting
that multiplying each $a_i$ by a non-zero constant does not modify the
performance of the likelihood ratio test. Hence we may normalize $A_n$ so that $\|\bs a_i\|=1$
for each~$i$ in the following. In this case, $\bs z$ is a zero
mean Gaussian random vector with covariance matrix
$$
\E_1(\bs z\bs z^T) = {A}_n^T{\Gamma}_n{A}_n + \sigma^2 {I}_k
$$
under hypothesis $H_1$, where $\Gamma_n=\E_1(\bs x\bs x^T)$ represents the $n×
n$ covariance matrix of vector $\bs x$.  Matrix $\Gamma_n$ is the $n× n$
Toeplitz matrix associated to the spectral density $f$ of process $x$, namely,
\begin{equation}\label{eq:gammafctf}
  \Gamma_n = T_n(f) = \left[\frac 1{2\pi}\int_{-\pi}^\pi f(\omega)e^{\rmi\omega(k-l)}\rmd\omega\right]_{1\leq k,l\leq n}\ .
\end{equation}
Under $H_0$, the covariance matrix of vector $\bs z$ simply coincides with
$\E_0(\bs z\bs z^T) = \sigma^2 {I}_k$. Using these remarks, it is
straightforward to show that 
\begin{equation}\label{eq:exprLLR}
  2 \,\L_{A_n} = k \log \sigma^2 +\frac{\| \bs z\|^2}{\sigma^2}-\log\det(A_n^T\Gamma_nA_n+\sigma^2I_k)
  -\bs z^T(A_n^T\Gamma_n A_n+\sigma^2I_k)^{-1}\bs z\:.
\end{equation}
In the sequel, we assume as usual that the threshold of the test, say
$\gamma_n$, is fixed in such a way that the probability of false alarm (PFA)
does not exceed a level $\alpha$ ($0<\alpha<1$), which reads
\begin{equation}
  \label{eq:pfa}
  \bP_0(\L_{A_n} > \gamma_n) \leq \alpha\:.
\end{equation}
We now analyze the miss probability of the above LLR test as a function of
$A_n$.

\subsection{Introduction to error exponents}

Let $P_M(\p;A_n)$ denote the miss probability of the LLR test with level $\p$
based on the observation $z_1,\dots,z_k$:
$$ 
P_M^n(\p;A_n)= \inf\; \bP_1(\L_{A_n} \leq \gamma_n)\;,
$$  
where the inf is taken over all threshold values $\gamma_n$ verifying the PFA
constraint~(\ref{eq:pfa}). The miss probability is generally the key metric to
characterize the performance of hypothesis tests. Unfortunately, an exact
expression of the miss probability as a function of $A_n$ is difficult to obtain
in the general case.  Following~\cite{Chen96}, we thus analyze the asymptotic
behaviour of the miss probability as the number of available observations tends
to infinity. More precisely, we study the asymptotic regime where both the
number of sensors $k$ and the number of observations $n$ per sensor tend to
infinity at the same rate:
\begin{equation}
  \label{eq:growth-k-assump}
  n\to\infty,\;\; k\to\infty,\;\; \frac{k}{n}\to c
\end{equation}
where $c\in(0,1)$. Any sequence of $n× k$ precoding matrices $\A=(A_n)_{n\geq
  0}$ will be refered to as a \emph{linear strategy}.  Loosely speaking, we will
prove that, at least for certain linear strategies of interest, the miss
probability behaves as
$$
P_M^n(\p;A_n) \simeq e^{-n K_{\p}(\A)}
$$
in the asymptotic regime~(\ref{eq:growth-k-assump}), where $K_{\p}(\A)$ is a
certain constant which depends on the linear strategy but, as a matter of fact,
does not depend on the level $\p$.  Such a constant is called the \emph{error
  exponent}. It is a key indicator of the way the power of the test is
influenced by the chosen linear strategy.  More formally, we define for each
$\A$, 
\begin{eqnarray}
  \underline K_{\p}(\A) & = & \liminf_{k\to\infty}-\frac1n\log P_M^n(\p;A_n)\;,
  \label{eq:expinf}\\
  \overline K_{\p}(\A) & = & \limsup_{k\to\infty}-\frac1n\log P_M^n(\p;A_n)\;,
  \label{eq:expsuf}
\end{eqnarray}
and we define the \emph{error exponent} of $\A$ as $K_{\p}(\A)=\underline
K_{\p}(\A)=\overline K_{\p}(\A)$ as soon as~(\ref{eq:expinf})
and~(\ref{eq:expsuf}) coincide.  In the sequel, our aim is therefore to
determine linear precoding strategies $\A$ having a large error exponent
$K_{\p}(\A)$ (and for which $K_{\p}(\A)$ is well-defined, of course). The
following Lemma (see \cite{Chen96}) provides a practical way to evaluate error
exponents.
\begin{lemma}[\cite{Chen96}]\label{lem:Chen}
  The following inequalities hold:
  \begin{eqnarray*}
    \underline K_{\p}(\A) &\geq&\sup\left\{t: \liminf_{n\to\infty}\bP_0\left[\frac1n\log\frac{p_0({\bs z})}{p_1({\bs z})}\leq
    t\right]<\p\right\}  \\
    \overline K_{\p}(\A) &\leq& \sup\left\{t:\limsup_{n\to\infty}\bP_0\left[\frac1n\log\frac{p_0({\bs z})}{p_1({\bs z})} \leq t\right]\leq\p\right\}\ .
  \end{eqnarray*}
  In particular if, under hypothesis $H_0$, $-n^{-1}\L_{A_n}$ converges in
  probability to a deterministic constant $\xi$, then $\overline K_{\p}(\A)=
  \underline K_{\p}(\A)= K_{\p}(\A)=\xi$ is necessary equal to this limit.
\end{lemma}

According to the above lemma, the asymptotic performance analysis of the LLR
test reduces to the characterization of the limit in probability of the
normalized LLR as $n\to\infty$, as soon as this limit exists. Moreover, in this
case, the error exponent $K_{\p}$ is independent from level~$\p$.

\subsection{Some families of precoders}

A natural approach to design relevant precoders would be to characterize the
linear strategies $\A$ which maximize the limit in probability (if it exists)
of the LLR $\L_{A_n}$ as $n,k\to\infty$. Ideally, this would lead to the
strategies with maximal error exponent.  Unfortunately, such a characterization
is difficult and would moreover lead to linear strategies which would deeply
depend on the spectral density $f$ of the signal to be detected. The practical
implementation of such optimal linear strategies would typically require to
communicate the whole function $f$ to each sensor via a feedback link from the
central processor. In this paper, we focus on the opposite on the case of
``dumb'' sensors \emph{i.e.}, sensors which are able to process information
with few or no knowledge of their mission or their environment.  More
precisely, we separately study the following linear strategies.
\begin{enumerate}
\item {\sl Random iid precoders:} A natural way to design dumb sensor
  networks is to select each sensor's precoder at random, independently from the
  network's mission.  Motivated by first by the simplicity of the approach and
  second by its widespread use in compressive sensing
  applications~\cite{can05}, we assume that matrix $A_n$ is one realization of a
  $n× k$ random matrix with zero mean iid entries.  In the case of random
  iid precoders, sensors are able to precode their information without any
  instructions from the fusion center.

\item {\sl Orthogonal precoders:} In this case, matrix $A_n$ is such that
  $A_n^TA_n=I_k$ \emph{i.e.}, the precoders $\bs a_1,\dots, \bs a_k$ are
  orthogonal.  Orthogonal precoders will reveal useful for the design of dumb
  but nevertheless efficient sensor networks.  Indeed, under this constraint,
  we are able to exhibit strategies that achieve the best error exponent.  In
  addition, when such precoders are used, we will show that a low complexity
  testing procedure can be implemented as an alternative to the costly
  Likelihood Ratio Test, without decreasing the error exponent.
\end{enumerate}

\section{Error Exponents}
\label{sec:errexp}

\subsection{Case of random iid precoders}
\label{sec:iid}

Before stating the main result of this subsection, remark that the performance
of the test is of course expected to depend on the covariance matrix $\Gamma_n$
of the signal to be detected. In particular, it is useful to recall some well
known results on the behaviour of the eigenvalues of $\Gamma_n$.  From classical
results on large Toeplitz matrices \cite{GrSz}, it is known that $\Gamma_n$ can
be approximated by a circulant matrix with eigenvalues $f(0), f(\frac
{2\pi}n),\dots,f(\frac {2\pi(n-1)}n)$. More precisely, for any Hermitian $n× n$ matrix $Q$,
we denote by $F_Q(t)=\frac {\# \{i, \,\lambda_i(Q)\leq t\}}{n}$ the
distribution function associated with the empirical distribution of the
eigenvalues $\lambda_1(Q),\dots,\lambda_n(Q)$ of $Q$ (the corresponding
probability measure is often refered to as the \emph{spectral measure} of $Q$).
Szegö's Theorem (\cite{GrSz}, p.64) states that, provided that Assumption~$\bf
A1$ holds, $F_{\Gamma_n}$ converges weakly to the distribution function $\Phi$
defined by:
\begin{equation}
\Phi(t) =\frac1{2\pi}\Leb\circ f^{-1}\big((-\infty,t]\big)\ ,
\label{eq:cdf}
\end{equation}
where we recall that $\Leb\circ f^{-1}$ is the measure which composes the
Lebesgue measure $\Leb$ on $[-\pi,\pi]$ with  $f^{-1}$ (the inverse image under
$f$). The error exponent merely
depends on the latter limiting spectral measure $\Phi$, as stated by the
following Theorem.

\begin{theorem}
\label{the:iid}
Suppose that~(\ref{eq:growth-k-assump}) holds for some $c\in(0,1)$ and assume
$\bf A1$. For each $n$, let $A_n = (A_{ij}^n)$ be a $n× k$ real random matrix such
that $A_{ij}^n$ for all $n,i,j$ are iid zero mean random variables with
finite second order moment. Consider any fixed
level $\p\in(0,1)$.  Then the linear strategy ${\cal A}=(A_n)_n$ admits an error
exponent $K_\p({\cal A})=K_{\mathrm{rnd}}(c)$ given~by:
\begin{equation}
  \label{eq:erriid}
  K_{\mathrm{rnd}}(c) = -c+\sigma^2c \beta-\frac c2 \log(\sigma^2\beta)+\frac 1{2}\int \log(1+ct\beta)\rmd\Phi(t)  \ ,
\end{equation}
where $\beta$ is the unique solution to the following equation:
\begin{equation}
  \sigma^2 = \frac 1\beta -\int \frac t{1+ct\beta}\rmd\Phi(t)\ .
\label{eq:beta}
\end{equation}
\end{theorem}
The proof is provided in Appendix~\ref{sec:proofiid}.

\subsection{Case of orthogonal precoders}
\label{sec:pfs}


We now focus on the case where $A_n^TA_n=I_k$. Our aim is first to prove that
among all orthogonal strategies, we may determine some that achieve the
maximum error exponent and second, to determine this maximum error exponent.
Results are provided below in Theorem~\ref{thm:optexp}. We first provide some
definitions along with some insights on the results.

Loosely speaking, it is easy to think of a relevant orthogonal strategy as
follows.  Focus on one given sensor $i=1\dots k$ for the sake of simplicity.
Under $H_0$, the received sequence $\bs y_i=\bs w_i$ corresponds to a white
Gaussian noise of variance $\sigma^2$.  Therefore the law of $z_i=\bs a_i^T\bs
y_i$ is ${\cal N}(0,\sigma^2)$. Under $H_1$, it is straightforward to show that
$z_i\sim {\cal N}(0,\bs a_i^T\Gamma_n\bs a_i+\sigma^2)$, where we recall that
$\Gamma_n=\E(\bs x\bs x^T)$ is the signal covariance matrix.  Clearly, the best
way for the sensor $i$ to discriminate $H_1$ versus $H_0$ is to chose the
precoder $\bs a_i$ which maximizes the variance $\bs a_i^T\Gamma_n\bs
a_i+\sigma^2$. This is achieved when $\bs a_i$ coincides with the eigenvector of
$\Gamma_n$ associated with the largest eigenvalue.  Generalizing this remark to
$k$ sensors, it is natural to introduce the strategy for which the $k$ precoders
$\bs a_1\dots\bs a_k$ coincide with the $k$ eigenvectors of $\Gamma_n$
associated with the largest eigenvalues. We shall refer to this strategy as the
{\sl Principal Component Strategy} (PCS).

\begin{definition}[principal components strategy
  (PCS)]\label{def:princ-comp-strat}
  Let $(\bs v_i^n)_{1\leq i\leq n}$ be the eigenvectors of $\Gamma_n$ and
  $(\lambda_i^n)_{1\leq i\leq n}$ be the corresponding eigenvalues, ordered in
  such a way that $\lambda_1\geq\lambda_2\dots\geq \lambda_n$. The principal
  component strategy $\V$ is defined as the sequence of $n× k$ matrices
  $V_n$ given by:
  $$
  V_n = \left[\bs v_1^n,\dots,\bs v_k^n\right]\ .
  $$
\end{definition}

As will be stated by Theorem~\ref{thm:optexp} below, PCS achieves the maximum
error exponent among all orthogonal strategies. Unfortunately, exact PCS might
be difficult to implement in a dumb sensor network, as each node needs to be
informed of a whole eigenvector of the covariance matrix $\Gamma_n$. This
requires involved cooperation between the nodes and the fusion center. In order
to reduce the amount of overhead in the network, we propose an alternative
strategy which turns out to achieve the same error exponent as PCS.  Let
$F_n=[F_n(i,j)]_{0\leq i,j\leq n-1}$ denote the $n× n$ real-valued orthogonal
Fourier basis matrix, that is $F_n=[\bs e_0^n,\dots,\bs e_{n-1}^n]$, where the
columns ${\bs e}_j^n$ are defined, up to a normalizing constant by  
$$
\left\{\begin{array}{lll}
 {\bs e}_0^n &\propto [1,\dots,1]^T &\\
 {\bs e}_j^n & \propto [\cos(2\pi ij / n)]_{i=0,\dots,n-1} & \text{for}\quad j=1,\dots,\lfloor n/2\rfloor\\
 {\bs e}_{n-j}^n &\propto [\sin(2\pi ij / n)]_{i=0,\dots,n-1} & \text{for}\quad j=1,\dots,\lfloor (n-1)/2\rfloor\ .
\end{array}\right.
$$
The main idea is to remark that for large $n$, the covariance matrix $\Gamma_n$
can be approximated by the matrix
\begin{equation}
  F_n\, \diag\left(f(0)\dots f(2\pi(n-1)/n)\right)\,F_n^T\ ,
\label{eq:diagFn}
\end{equation}
see~\cite{GrSz,widom} for more details. As a consequence, it seems reasonable to
propose a strategy inspired of PCS, only substituing the above
matrix~(\ref{eq:diagFn}) with the true covariance matrix $\Gamma_n$. This leads
to the following definition.

\begin{definition}[principal frequencies strategy (PFS)]\label{def:princ-freq-strat}
  For each $n$, denote by $(j_1^n,\dots,j_n^n)$ any permutation of $\{0,1,
  \dots, n-1\}$ such that $f(2\pi j_1^n/n)\geq \dots \geq f(2\pi j_n^n/n)$.  The
  principal frequencies strategy ${\W}$ is defined as the sequence of $n× k$
  matrices $W_n$ given by:
  \begin{equation}
    \label{eq:matricePFS}
    W_n = \left[\bs e_{j_1^n}^n,\dots,\bs e_{j_k^n}^n\right]\ ,
  \end{equation}
  where $\bs e_1^n,\dots,\bs e_n^n$ are the columns of matrix $F_n$.
\end{definition}
Note that PFS only requires to transmit one of the $k$ indices $j_1^n\dots
j_k^n$ corresponding to the principal frequencies of $f$ to each sensor.
In return, each sensor
$i$ computes the scalar product between the $j_i^n$th column of Fourier matrix
$F_n$ and its received sequence $\bs y_i$. In other words, it computes the value
of the (real) periodogram of $\bs y_i$ at frequency~$2\pi j_i^n/n$.  The
following result proves furthermore that both PCS and PFS achieve the best error
exponent among all orthogonal strategies.

For any $c\in(0,1)$ denote by $\Delta_c$ the following set of frequencies:
\begin{equation}
  \label{eq:Delta}
  \Delta_c = \left\{\omega\in(-\pi,\pi)\,:\,\Phi\circ f(\omega)\geq 1-c\right\}\:.
\end{equation}
It is worth noting that the Lebesgue measure of $\Delta_c$ is equal to $2\pi c$
(see Lemma~\ref{lem:DeltaC}).

\begin{theorem}\label{thm:optexp}
 Suppose that~(\ref{eq:growth-k-assump}) holds for some $c\in(0,1)$.
 and assume $\bf A1$ and $\bf A2$.  Let $\cal V$ and $\cal W$
  respectively denote the PCS and PFS as defined above.  For any $\p\in(0,1)$,
  the error exponents $K_\p(\V)$ and $K_\p(\W)$ associated with $\V$ and $\W$
  exist, and are such that $K_\p(\V)=K_\p(\W)=K_{\mathrm{orth}}(c)$ where
  \begin{equation}
    \label{eq:WF-err-exp}
    K_{\mathrm{orth}}(c) = \frac1{2\pi}\int_{\Delta_c}D\left(\N(0,\sigma^2)\,||\,\N(0,f(\omega)+\sigma^2)\right)\rmd\omega\;,
  \end{equation}
  where $D$ denotes the Kullback-Leibler contrast.
  Moreover, for any orthogonal strategy $\A$, 
   \begin{equation}
    \label{eq:WF-err-exp-LB}
  \overline K_{\p}(\A)\leq K_{\mathrm{orth}}(c)\ . 
\end{equation}
\end{theorem}

The proof is provided in Appendix~\ref{sec:proof-theor-optexp}.
Let us briefly comment the best error exponent
formula~(\ref{eq:WF-err-exp}). First we recall that for any $\sigma_1^2,\sigma_2^2>0$,
$$
D\left(\N(0,\sigma_1^2)\,||\,\N(0,\sigma_2^2)\right)=-\frac12\left[\log\frac{\sigma_1^2}{\sigma_2^2}+1-\frac{\sigma_1^2}{\sigma_2^2}\right]\;,
$$
which is increasing as $\sigma_1/\sigma_2$ gets away from 1 from above or
below. 
Since $\Phi$ is nondecreasing, we see that the frequencies $\omega$
lying in $\Delta_c$ are those that maximize
$D\left(\N(0,\sigma^2)\,||\,\N(0,f(\omega)+\sigma^2)\right)$ in $[-\pi,\pi]$.
Thus $K_{\mathrm{orth}}(c)$ can be interpreted as some distance between the two
spectral densities $\sigma^2$ (corresponding to $H_0$) and $f+\sigma^2$
(corresponding to $H_1$) restricted to a set of frequencies where these two
spectral densities are the furthest apart.


\subsection{Illustration and comparisons}\label{sec:numee}

Error exponents $K_{\mathrm{orth}}$ and $K_{\mathrm{iid}}$ defined in sections
\ref{sec:pfs} and \ref{sec:iid} depends on the following parameters: the
spectral density $f$, the noise level $\sigma$, along with the sensors growth
ratio $c$. When using the orthogonal strategy, one can expect that the more
peaky $f$ is, the more efficient the compression will be. That is, by using
only a few sensors configured at the peak frequencies, one will get a
attractive exponent error. This should also lead to a sharp increase of the
error exponent curve $K_{\mathrm{orth}}(c)$ for small $c$. On the contrary, when $f$ is nearly flat
(with a small range of values), there are no priviledged frequencies for the
sensors to forward and the error exponent should increase slowly 
as $c$ gets larger. Let us illustrate these intuitive arguments with numerical
experiments.
We consider two spectral densities corresponding to ARMA processes. The
corresponding plots are depicted in Fig. \ref{fig:spectral_dens}.
$$
\begin{array}{rcl}
  f_1(\omega) & = & s_1^2\left|\frac{1 - \frac{1}{2}\exp(i\omega) + \frac{1}{4}\exp(2i\omega)}{1 - \frac{1}{2}\exp(i\omega) - \frac{1}{5}\exp(2i\omega) - \frac{1}{10}\exp(3i\omega)}\right|^2\\
  f_2(\omega) & = & s_2^2\left|\frac{1 + \frac{7}{10}\exp(i\omega) - \frac{1}{5}\exp(2i\omega)}{1 + \frac{2}{5}\exp(i\omega) - \frac{3}{10}\exp(2i\omega)}\right|^2\\
\end{array}
$$

\begin{figure}[ht!]
  \centering
  $$
  \begin{array}{cc}
    \includegraphics[width=.45\linewidth]{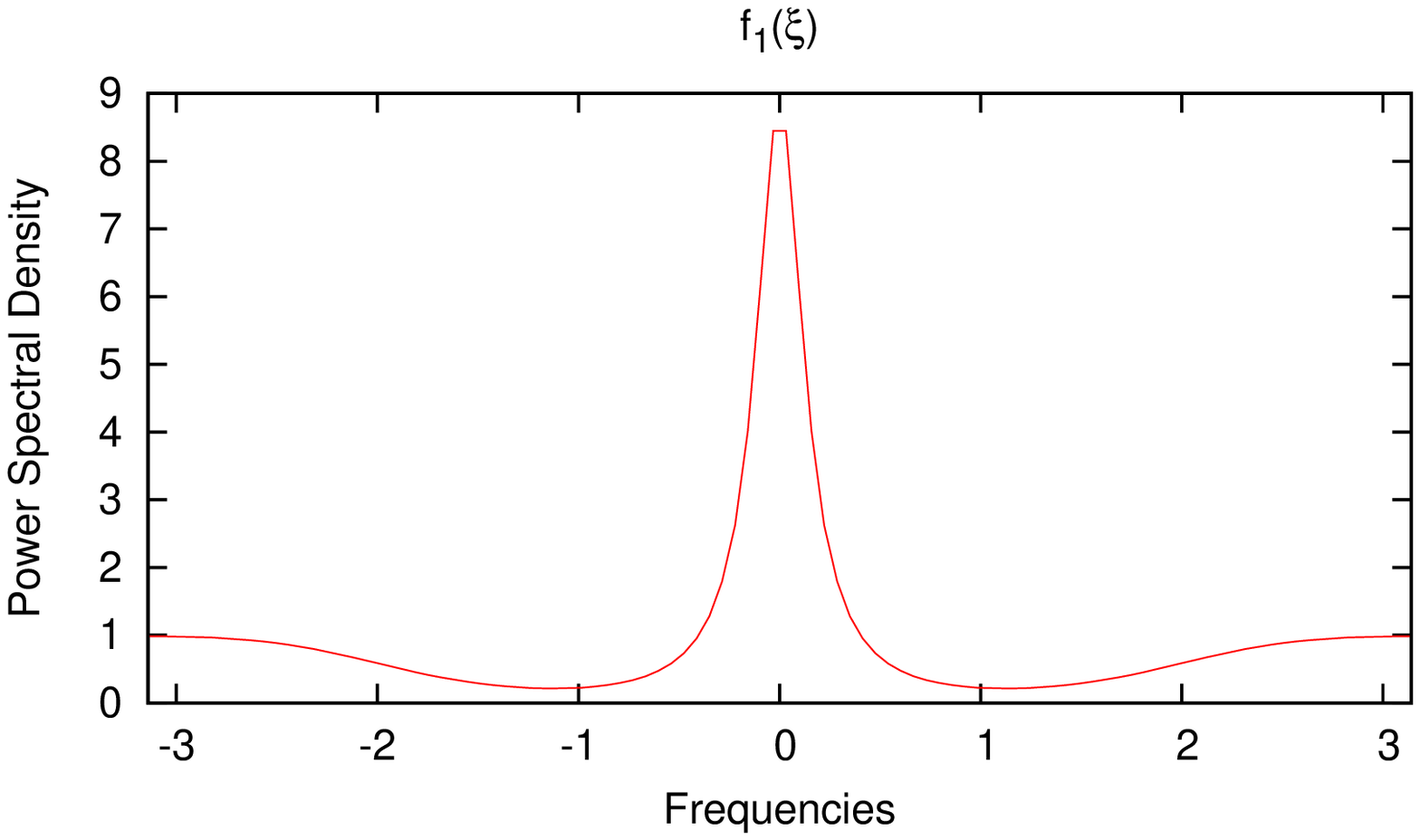}&
    \includegraphics[width=.45\linewidth]{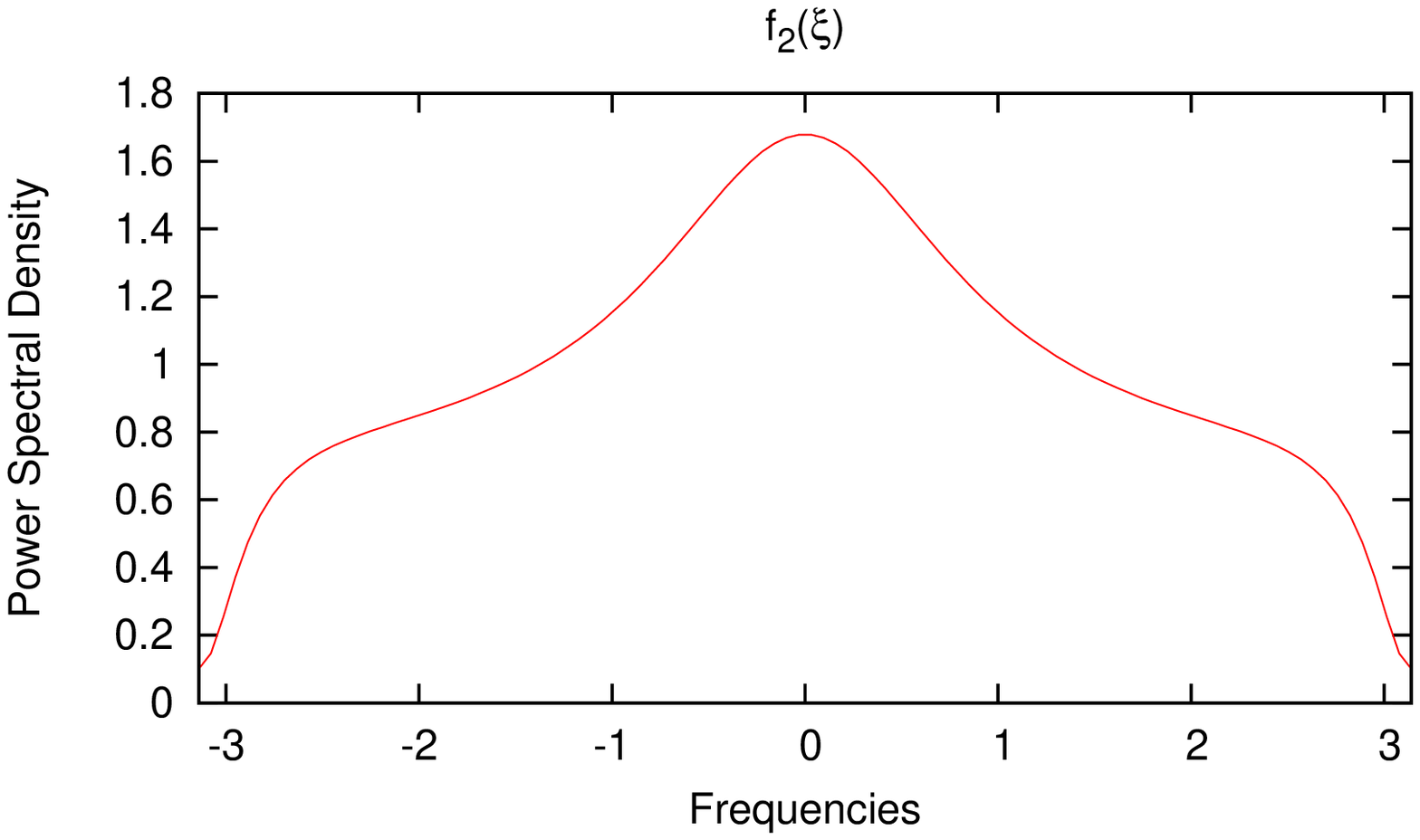}\\
  \end{array}
  $$
  \caption{Left: Spectral density $f_1$ with $s_1$ adjusted such that
    $\frac{1}{2\pi}\int_{-\pi}^{\pi}f_1 = 1$. Right: Spectral density $f_2$ with
    $s_2$ adjusted such that $\frac{1}{2\pi}\int_{-\pi}^{\pi}f_2 = 1$. One can
    notice that $f_1$ has a sharp peak while $f_2$ takes its values in much
    smaller range.}
  \label{fig:spectral_dens}
\end{figure}

Fig.~\ref{fig:ee_curves} represents $K_{\mathrm{iid}}(c)$ and $K_{\mathrm{orth}}(c)$ for
$\sigma=1$. For comparison, we also plotted another error exponent curve,
corresponding to an orthogonal, yet suboptimal strategy which uses precoding matrices
$A_n=[I_k\;\;0]$. This strategy amounts to keep
only the first $k$ values of the signal, independently of $f$. It is
straightforward to prove that the corresponding error exponent writes
$K_{\mathrm{fst}}(c) = c\cdot K_{\mathrm{orth}}(1)$.  

\begin{figure}
\centering
$$
\begin{array}{cc}
  \includegraphics[width=.45\linewidth]{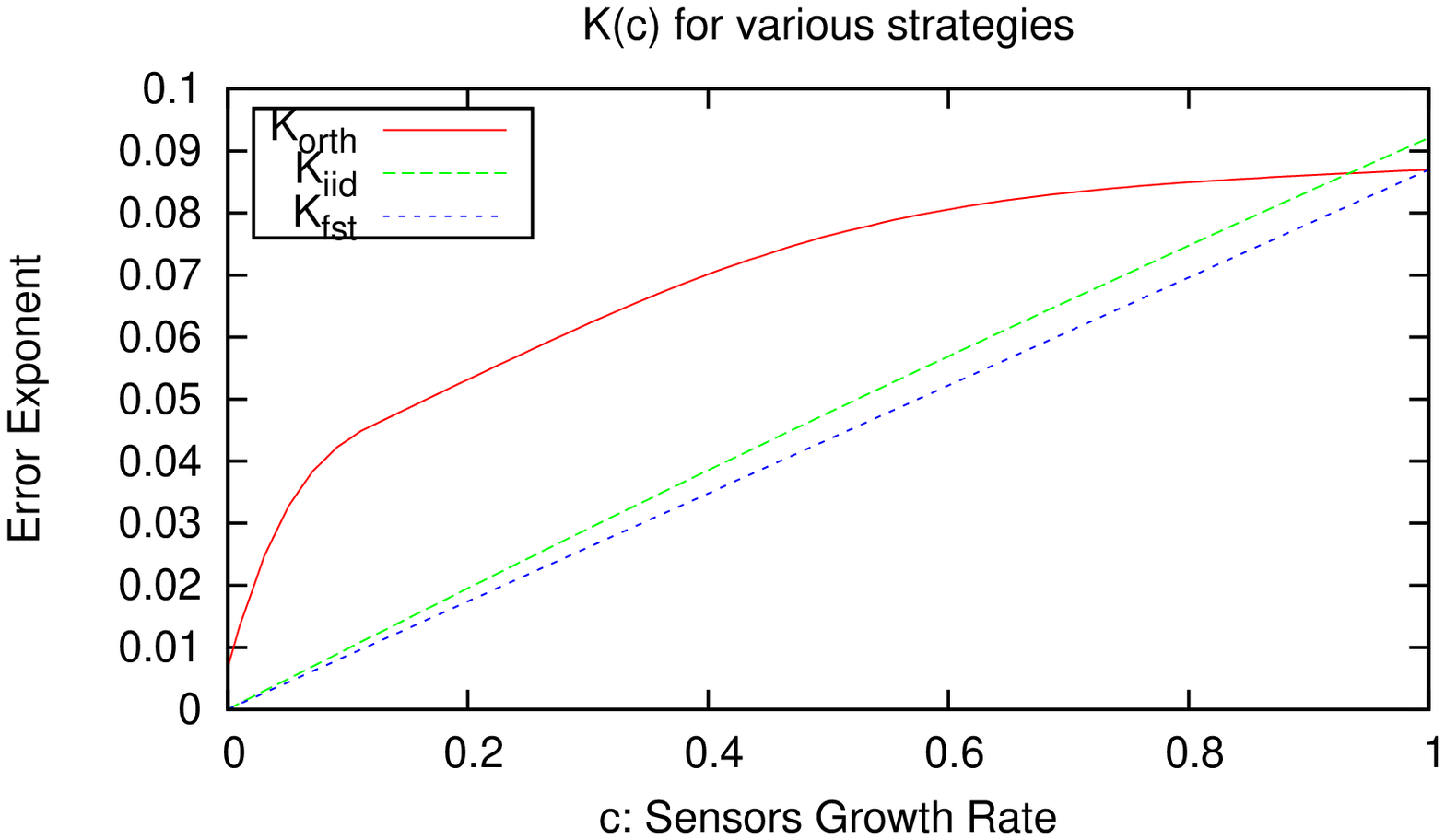}&
  \includegraphics[width=.45\linewidth]{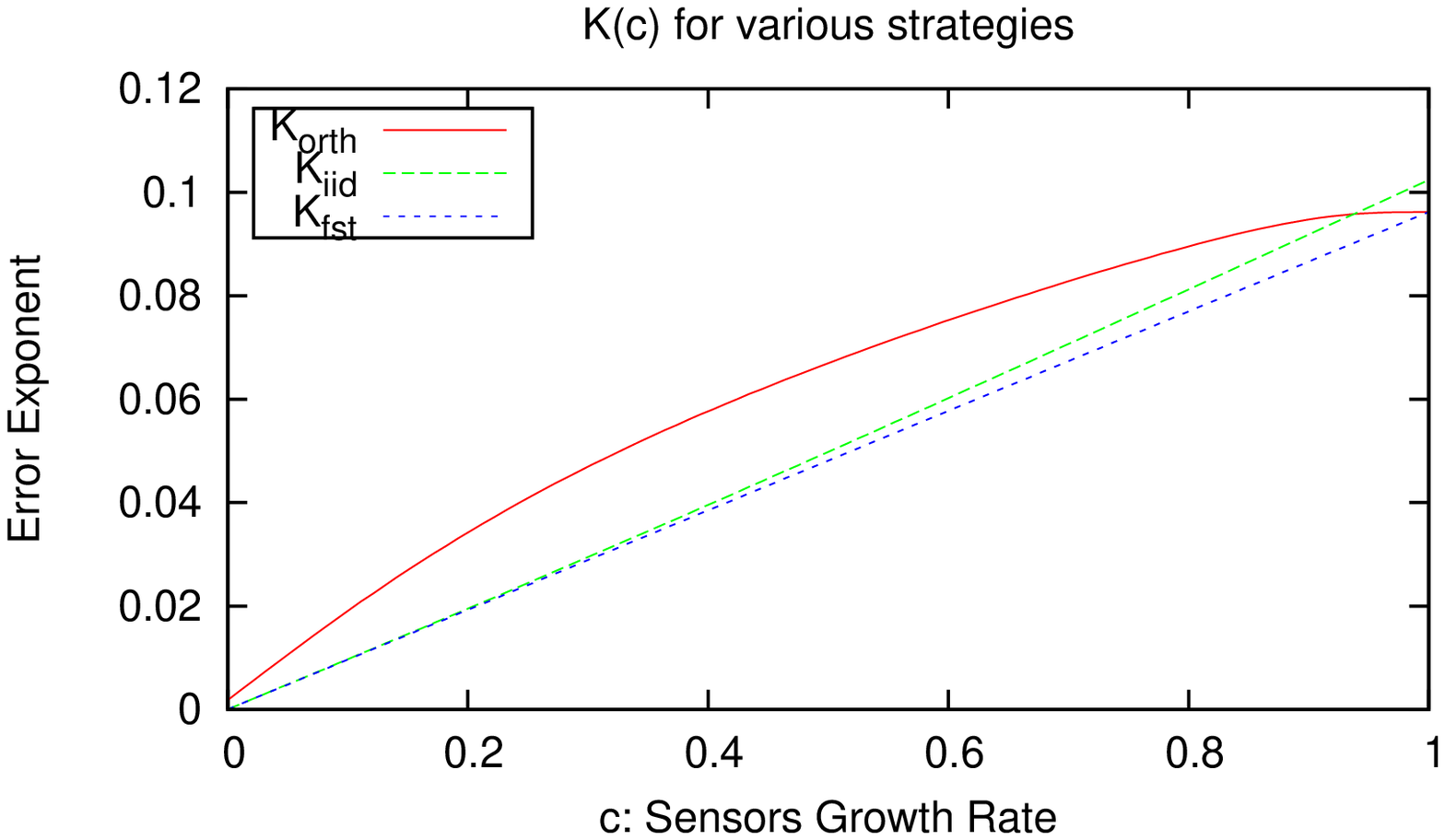}
\end{array}
$$
\caption{Error Exponents $K_{\mathrm{iid}}(c)$, $K_{\mathrm{orth}}(c)$ and
  $K_{\mathrm{fst}}(c)$ as functions of the growth ratio $c=\lim k/n$, for
  spectral density functions $f_1$ (left) and $f_2$ (right).}\label{fig:ee_curves}
\end{figure}

One can notice several numerical facts on Fig.~\ref{fig:ee_curves}. First, as
expected from section~\ref{sec:pfs}, $K_{\mathrm{fst}}$ is always below
$K_{\mathrm{orth}}$. Remark that, as expected, $K_{\mathrm{orth}}$ has a
sharper increase near $c=0$ when used with $f_1$ than when used with $f_2$. The
fact that the random iid strategy seems to behave better for $c$ close to $1$
is more surprising but it reveals the following interesting fact: in some
circumstances, a \emph{non-orthogonal strategy may outperform an optimal
  orthogonal strategy}. Let us try to interpret this result. When setting up
the sensor network design, one faces a tradeoff. Either use an extra sensor
over the same frequency that the previous sensor in order to denoise their
common measurment, or use this extra sensor over a new frequency in order to
discover another part of the spectral density $f$. At small levels of noise, it
is always more interesting to discover $f$ at new frequencies than to denoise
ones already used by other sensors; indicating that the orthogonal strategy is
always the best. But for high levels of noise, it may become more efficient to
repeat (and thus denoise) key frequencies than to discover less important
ones. To support this claim, we refer to Fig.~\ref{fig:ee_curves_sigma} where
we chose two levels of noise, one that is larger ($\sigma^2 = 4$) than the one
used in Fig.~\ref{fig:ee_curves}, and one that is smaller ($\sigma^2 =
.5$). One can see that when $\sigma^2 = .5$, the best orthogonal strategy
outperforms the two others, whereas for $\sigma^2=4$ the upcrossing of
$K_{\mathrm{iid}}$ over $K_{\mathrm{orth}}$ near $c=1$ is more important than on
Fig.~\ref{fig:ee_curves}. The same conclusions hold for both spectral densities
$f_1$ and $f_2$.

\begin{figure}
\centering
$$
\begin{array}{cc}
  \includegraphics[width=.45\linewidth]{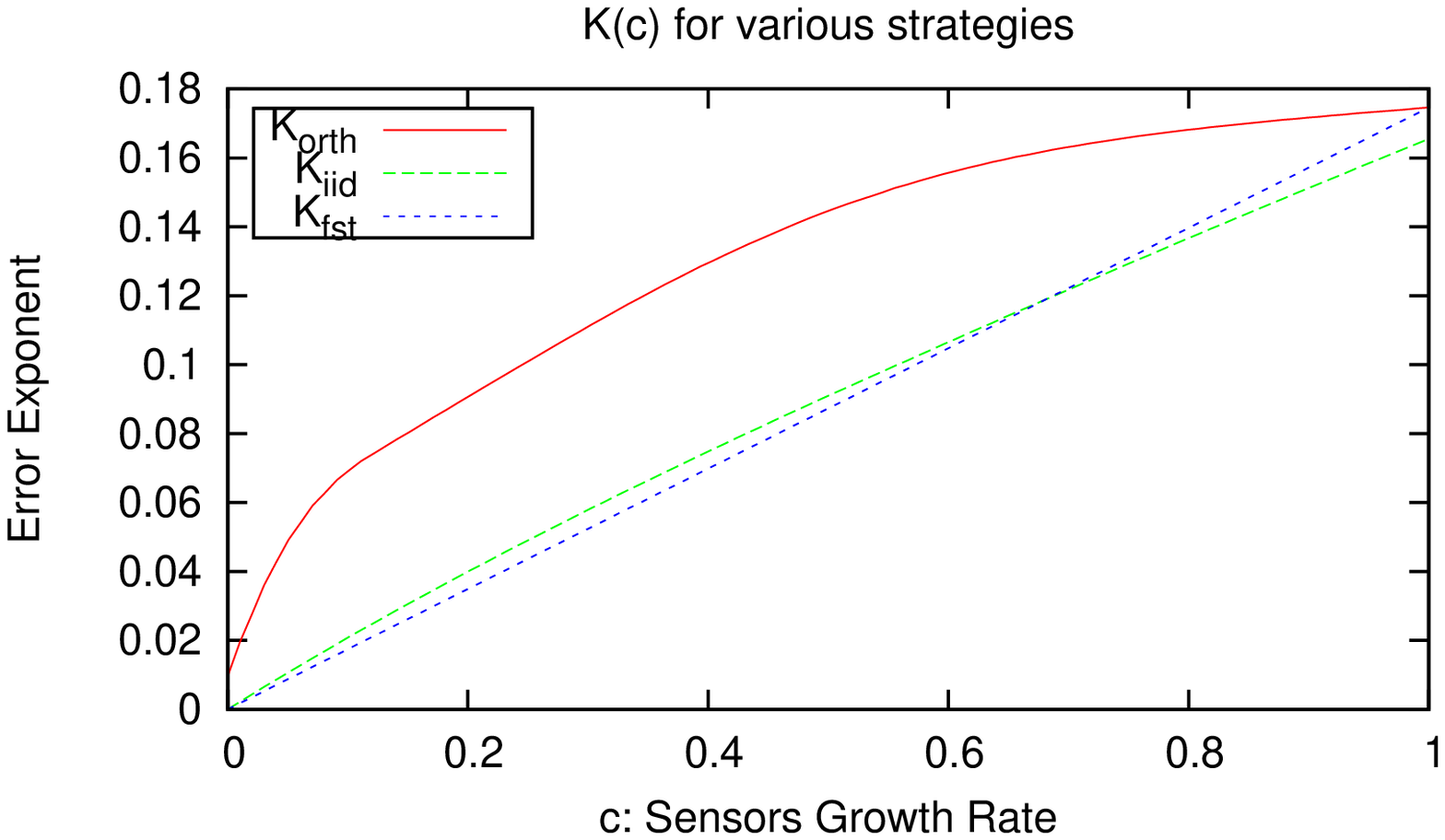}&
  \includegraphics[width=.45\linewidth]{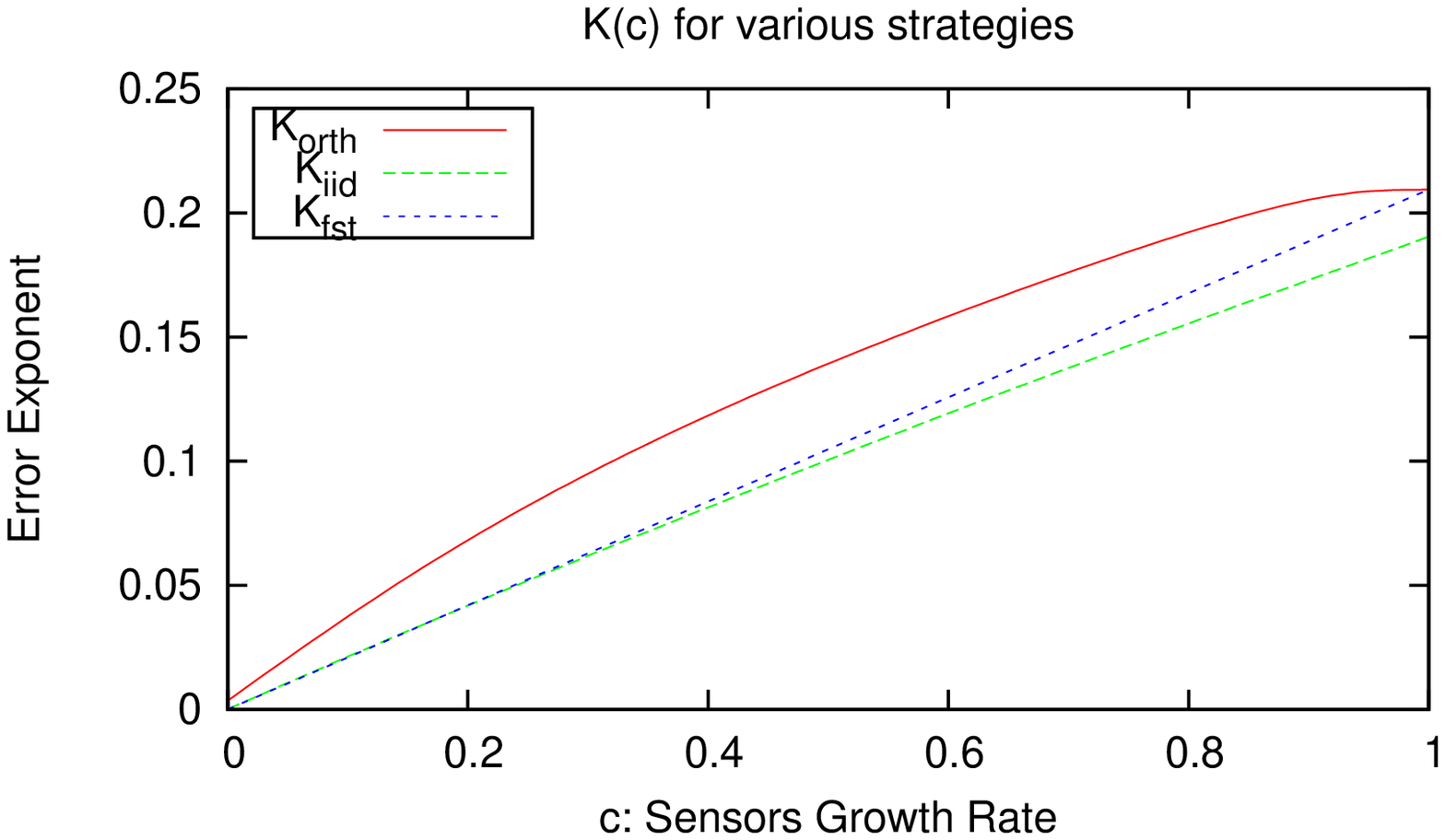}\\
  \includegraphics[width=.45\linewidth]{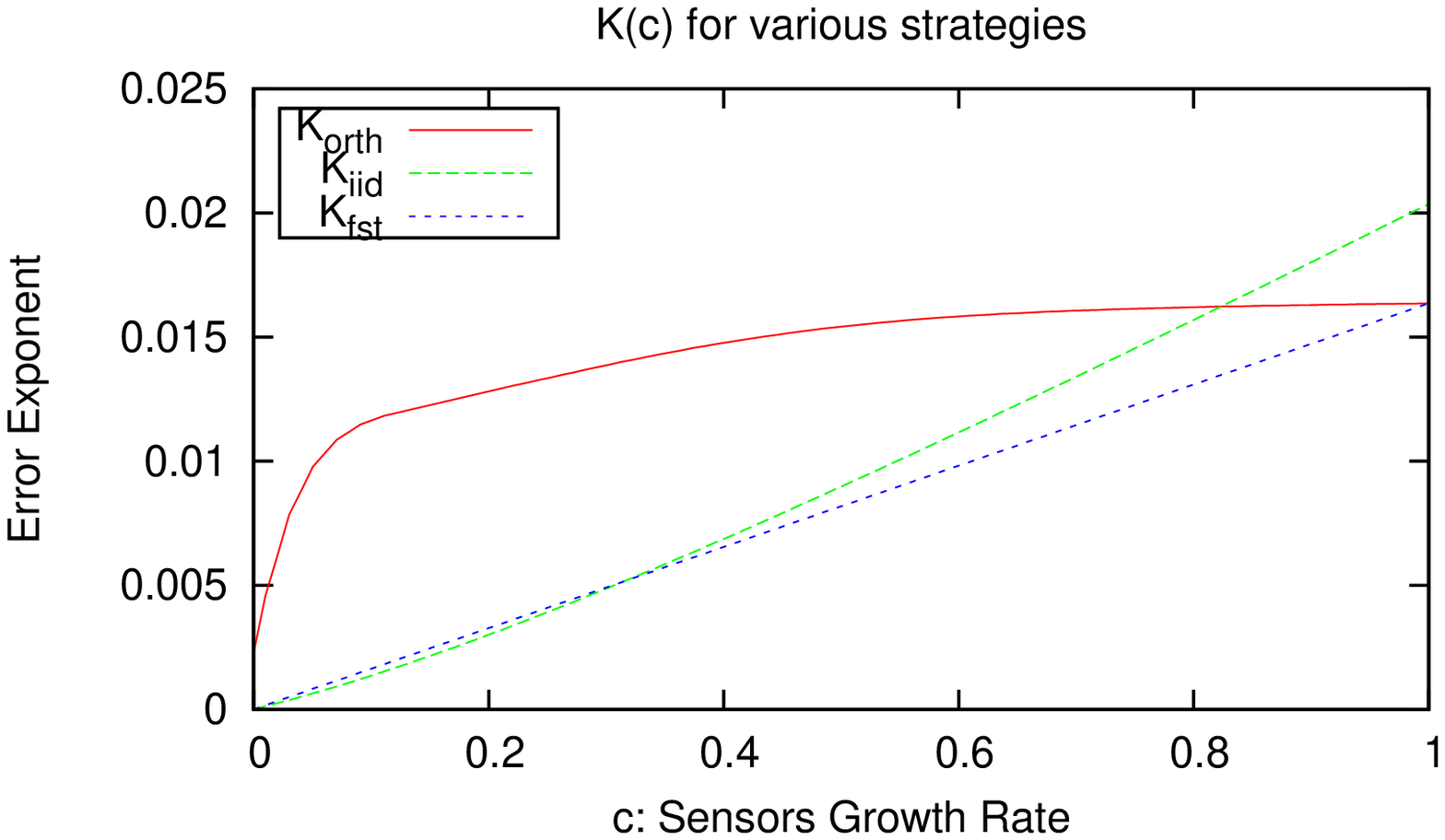}&
  \includegraphics[width=.45\linewidth]{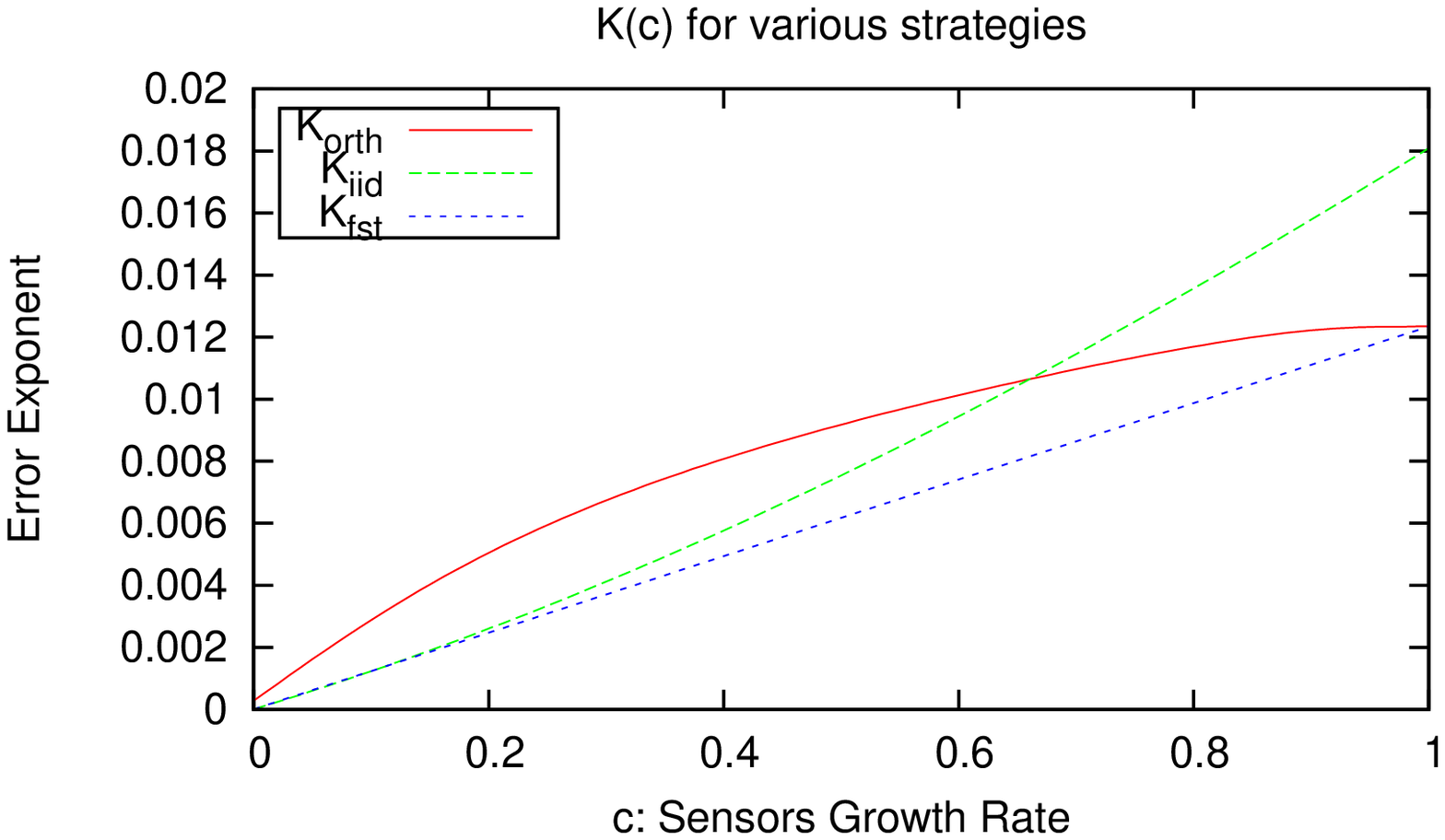}\\
\end{array}
$$
\caption{Error Exponent curves for spectral 
  density functions $f_1$ (left) and $f_2$ (right) with
  $\sigma^2 = 1/2$ (top) and $\sigma^2 = 4$ (bottom).}\label{fig:ee_curves_sigma}
\end{figure}

\section{A PFS-based low complexity test}\label{sec:approxPFS}

Results of the previous section indicate that the principal frequencies strategy
is a good candidate for implementation in dumb sensor networks. Indeed it
requires only few cooperation between the nodes and the fusion center, and is
attractive from an error exponent perspective.  In this section, we prove
furthermore that when PFS is used, then a test procedure can be proposed
which is much less complex than the LRT, and which achieves nevertheless the
same error exponent. 

We assume throughout this section that PFS is used \emph{i.e.}, each precoding
matrix is given by $A_n = W_n$ where $W_n$ is defined by~(\ref{eq:matricePFS}).

\subsection{A low complexity test}

Recall that the LRT rejects the null hypothesis when the LLR~(\ref{eq:exprLLR})
is above a threshold. As the terms $k\log\sigma^2$ and $\log\det
(A_n^T\Gamma_nA_n+\sigma^2I_k)$ are constant w.r.t. the observation $\bs z$, it
is clear that the LRT reduces to the test which rejects the null hyopthesis for
large values of the statistics:
\begin{equation}
\label{eq:equivLLRT}
\frac{\| \bs z\|^2}{\sigma^2}-\bs z^T(A_n^T\Gamma_n A_n+\sigma^2I_k)^{-1}\bs z \ .
\end{equation}
Unfortunately, the evaluation of the above statistics is computationally
demanding as $k$ gets larger, since it requires the inversion of the $k×
k$ matrix 
$A_n^T\Gamma_n A_n+\sigma^2I_k$. In order to avoid this, we propose to
replace matrix $\Gamma_n$ in~(\ref{eq:equivLLRT}) with its circulant
approximation given by~(\ref{eq:diagFn}).  In other words, product
$A_n^T\Gamma_n A_n$ is replaced by:
$$
A_n^TF_n\, \diag\left(f(0)\dots f(2\pi(n-1)/n)\right)\,F_n^TA_n =
\diag\left(f(2\pi j_1^n/n)\dots f(2\pi j_k^n/n)\right)\ .
$$
This leads directly to the following procedure.

{\bf PFS low complexity (PFSLC) Test:} Reject hypothesis $H_0$ when the statistics
${\cal T}_n$ defined~by:
\begin{equation}
  \label{eq:simpleStat}
  \mathcal{T}_n=\sum_{\ell=1}^{k} |z_{\ell}|^2 \left(\frac1{\sigma^2}-\frac1{\sigma^2+f(2\pi j_{\ell}^n/n)}\right)\;,
\end{equation}
is larger than a threshold.
\smallskip

Although this statistics cannot give rise to a better test than the LRT, its
numerical simplicity makes it worth to be considered. In the next paragraph, we
study the performance of the test and we prove that it performs as well as the
LRT in terms of error exponent.

\subsection{Asymptotic optimality of the PFSLC test}

As the statistics~(\ref{eq:simpleStat}) is no longer a likelihood ratio,
Lemma~\ref{lem:Chen} cannot be used to evaluate the error exponent associated
with the test~(\ref{eq:simpleStat}). Instead, we must resort to arguments of
large deviations theory. Specifically, we shall study the large deviation
behaviour of the test associated to this statistic, that is the limit of
$-n^{-1}\log\bP_1(\mathcal{T}_n\leq\eta_n(\p))$ where $\eta_n(\p)$ is the
$(1-\p)$-quantile of the statistic $\mathcal{T}_n$ under $H_0$,
$\bP_0(\mathcal{T}_n>\eta_n(\p))=\p$. Under mild assumptions, we show below that
this limit is given by the error exponent of the PFS. Hence, as far as error
exponents are considered, there is no loss in the performance in using the
statistic $\mathcal{T}_n$ defined in~(\ref{eq:simpleStat}) rather than
the likelihood ratio.


\begin{theorem}\label{thm:simpleStat}
  Assume that $\bf A1$ and $\bf A2$ hold true. For any level $\p\in(0,1)$, the statistics
  $\mathcal{T}_n$ defined in~(\ref{eq:simpleStat}) satisfies the following
  property. For $\eta_n(\p)$ such that $\bP_0(\mathcal{T}_n>\eta_n(\p))=\p$,
  the miss probability $\bP_1(\mathcal{T}_n\leq\eta_n(\p))$ satisfies
  $$
  \lim_{n\to\infty}-\frac1n\log\bP_1(\mathcal{T}_n\leq\eta_n(\p))= K_{\mathrm{orth}}
  \; .
  $$
\end{theorem}
The proof of this result is provided in Appendix~\ref{sec:proof-theor-statSimple}.


\section{Simulations}
\label{sec:simus}

The error exponent theory is inherently asymptotic. In this section we provide
numerical experiments to analyze the performance of the PFS on simulated data
for finite $n$ since we have already proved that the error exponent curve is the
same. The point here is to test how well the error exponent theory is relevant
for finite $n$.

We use the same spectral density functions $f_1$ and $f_2$ as in section
\ref{sec:numee}, whose error exponents are displayed in Fig.
\ref{fig:ee_curves}. We now compare, for a couple of values for $c$, the finite
sample performances of the LRT with the iid, PFS and PCS Strategies by using
their empirical Receiver Operating Characteristic (ROC) curves.  When the PFS is
used, we also consider the PFSLC test of section \ref{sec:approxPFS}.  We have
shown that the LRT with the PFS or the PCS and the PFSLC test share the same
error exponent curve. How well this measure of the performance impacts the whole
ROC curves at finite samples is displayed in Fig. \ref{fig:stratROCcomp}. It
turns out that the PCS, the PFS and the PFSLC have similar ROC curves, as
indicated by the error exponent analysis. One can also notice the good
performance of the PFS, PFSLC and PCS when $c=.1$, $\sigma=1$ and $n=100$ for
$f_1$, which confirms the conclusions drawn from the error exponent curves in
Fig.~\ref{fig:ee_curves}. For $c=.9$, $\sigma=2$ and $n=100$, one can notice
that error exponent curves also provide a good prediction: the iid strategy
slightly outperforms the PFS, PCS and PFSLC.

\begin{figure}[ht]
  $$
  \begin{array}{cc}
    \includegraphics[width=.45\linewidth]{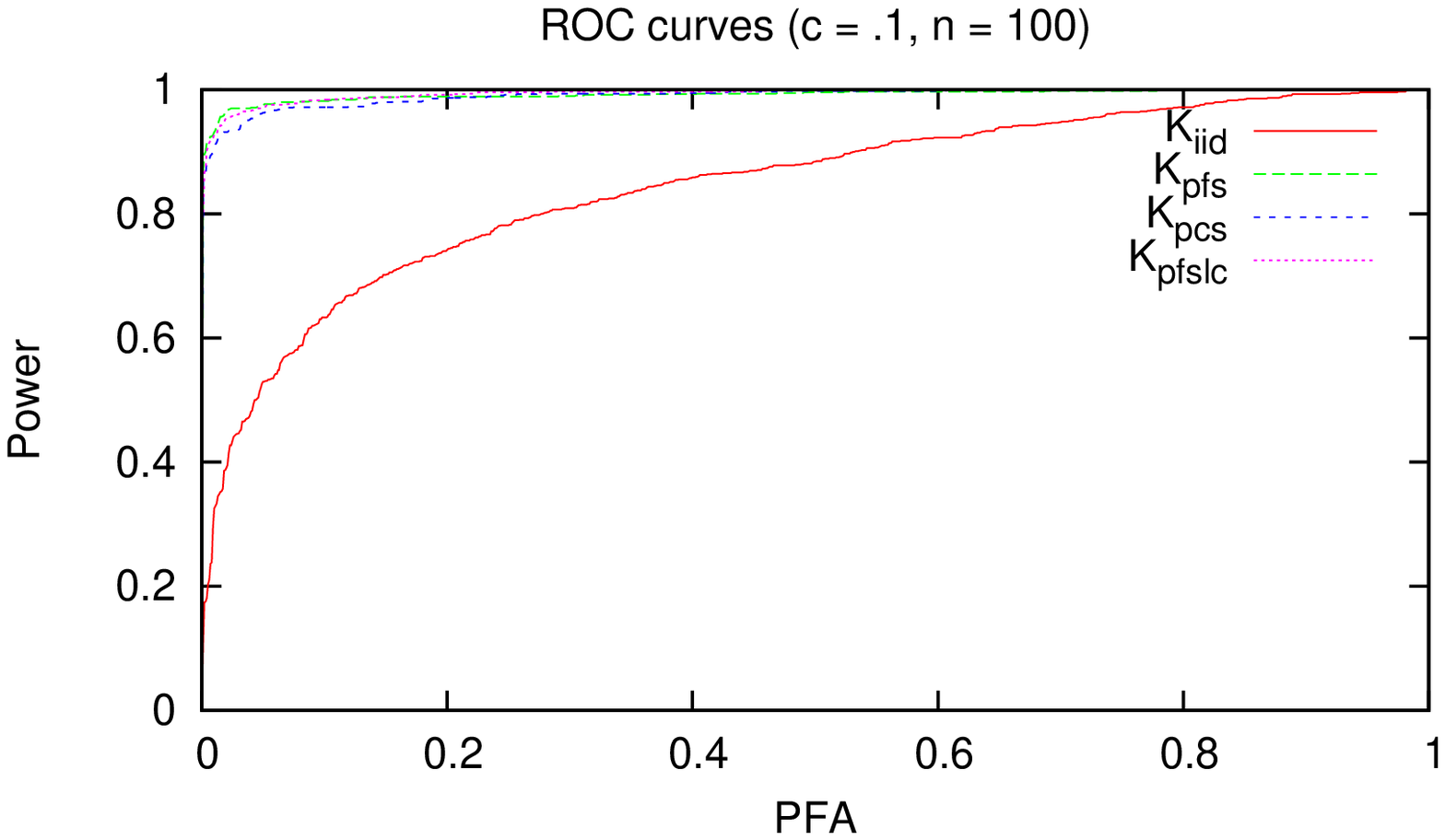}&
    \includegraphics[width=.45\linewidth]{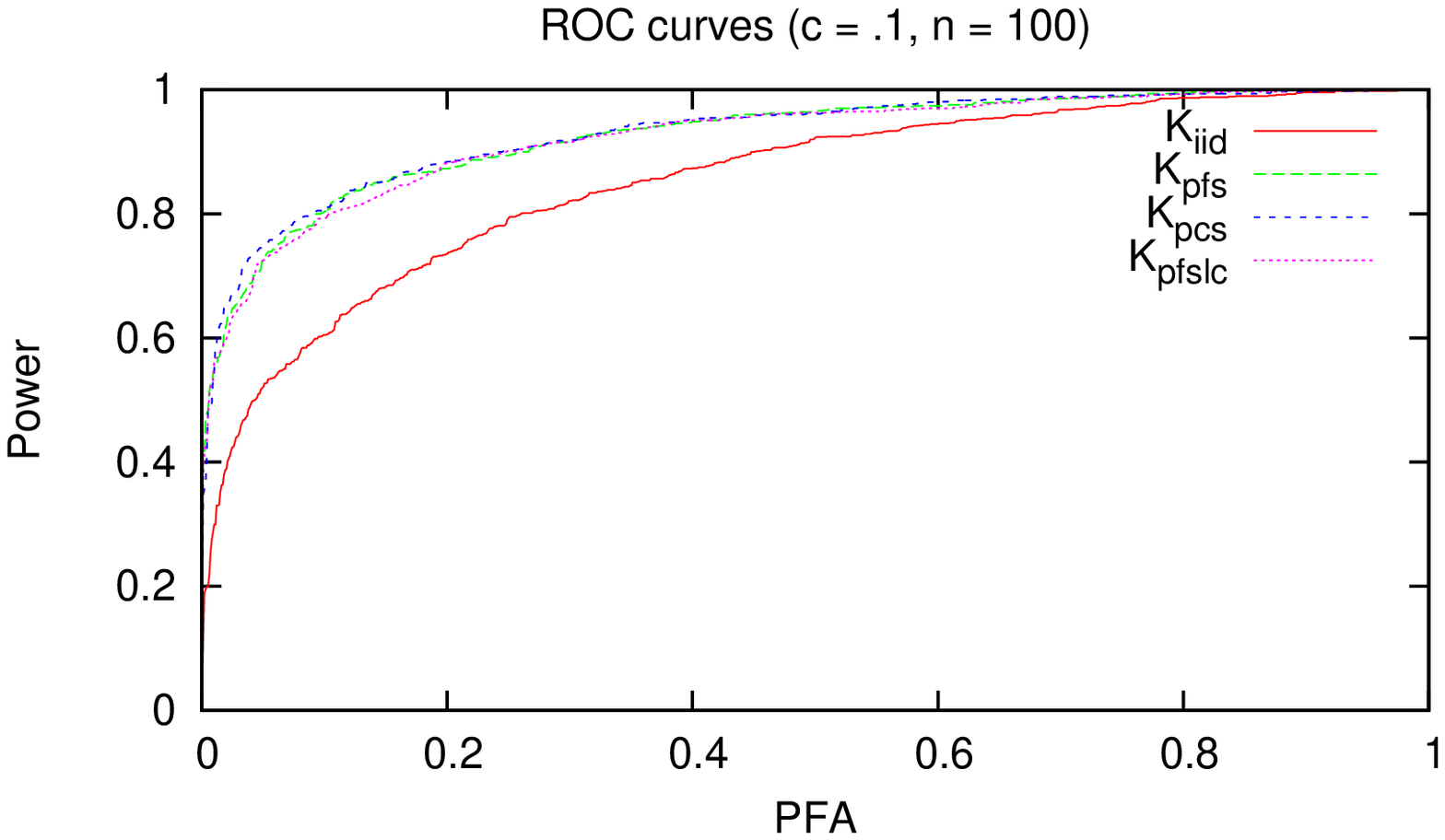}\\
    \includegraphics[width=.45\linewidth]{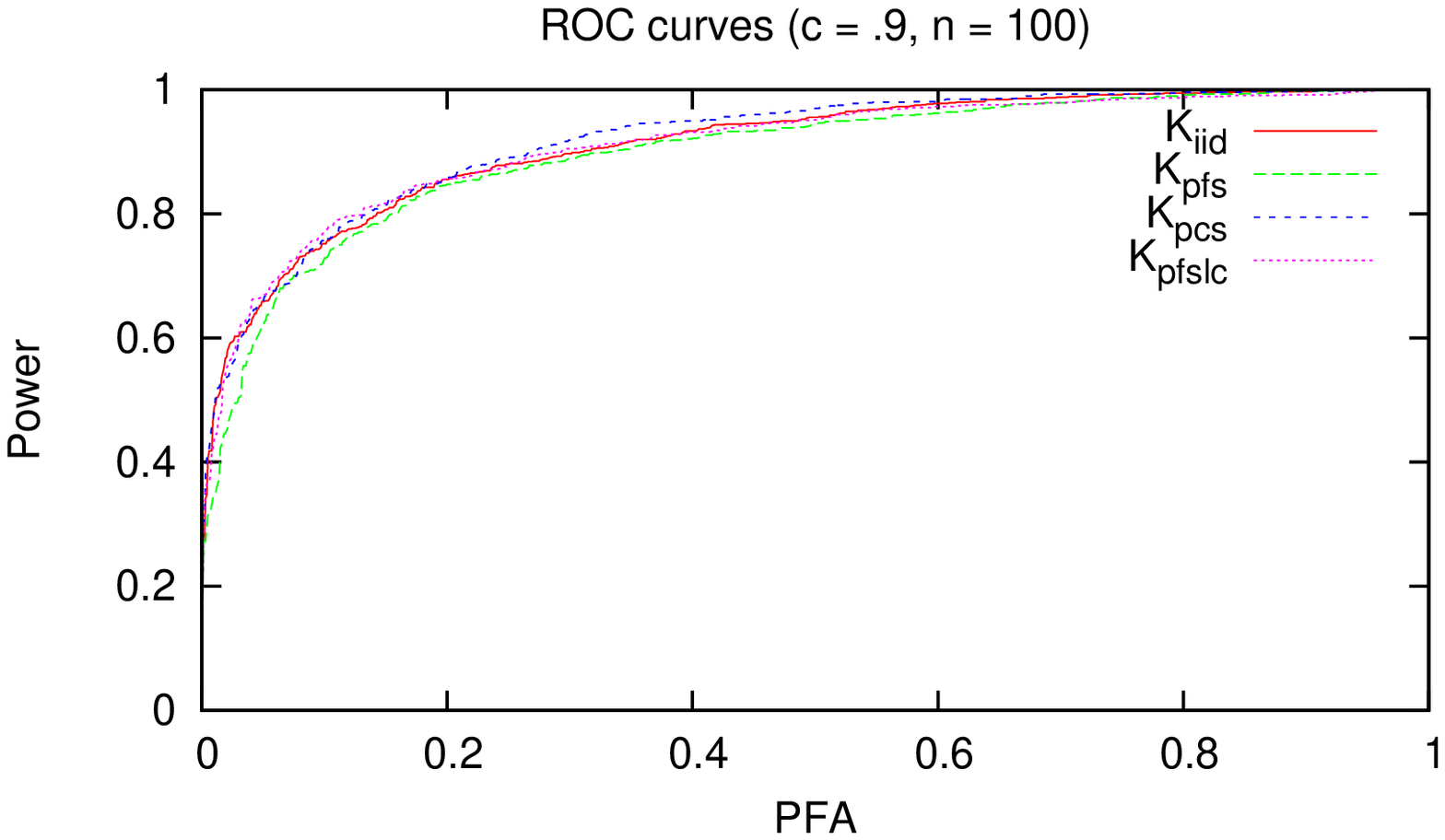}&
    \includegraphics[width=.45\linewidth]{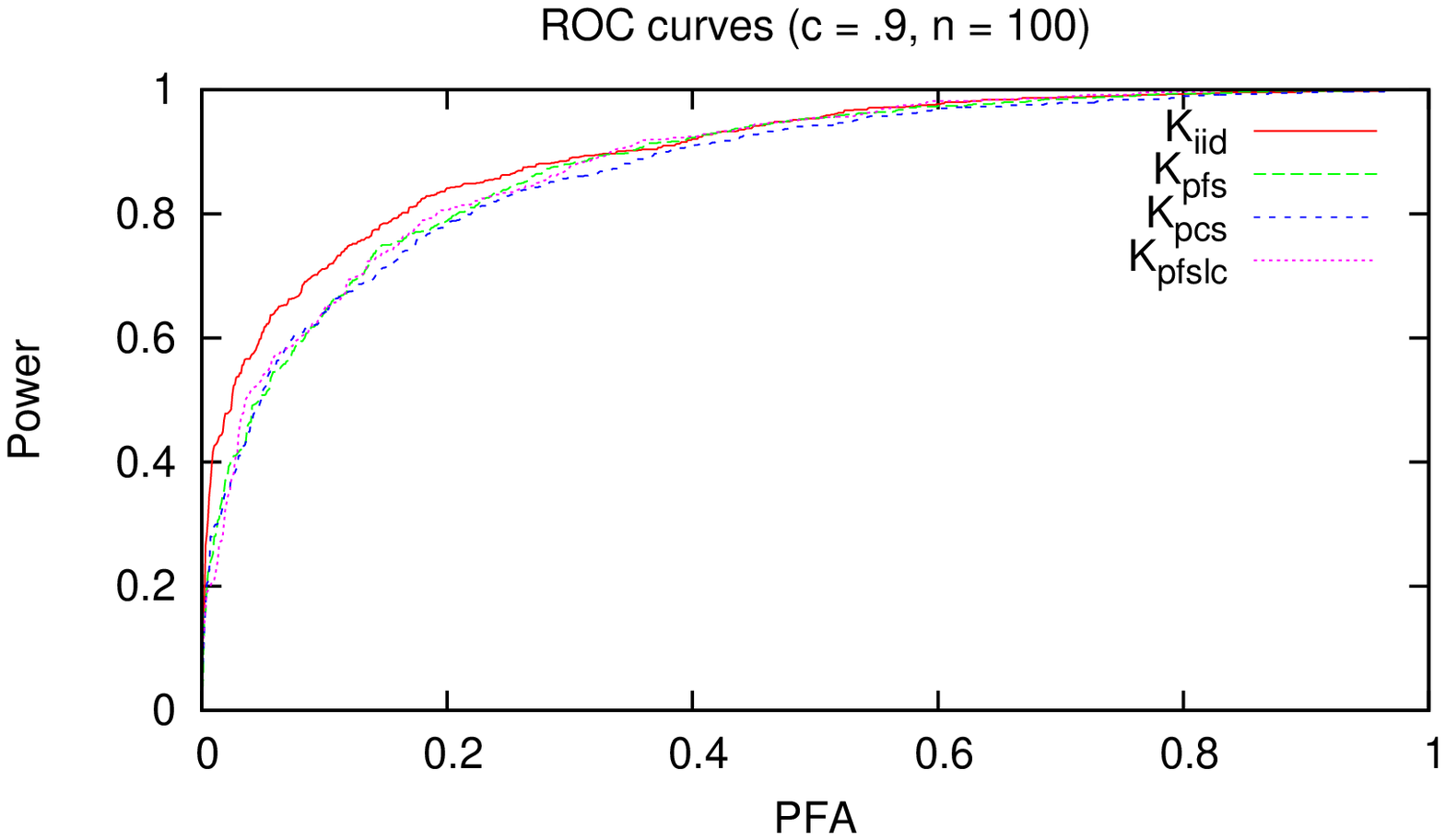}\\
  \end{array}
  $$
  \caption{ROC curves associated to the PFS, PCS, PFSLC and Random iid
  Strategy for spectral density functions $f_1$ (left) and $f_2$ (right). Top: $c=.1$ and
	$\sigma=1$. Bottom $c=.9$ and $\sigma=2$. As predicted by the error exponents
	curves, iid strategy is less efficient when $\sigma=1$ but slightly more
	efficient when $\sigma=2$ for large values of $c$.}\label{fig:stratROCcomp}
\end{figure}

\section{Conclusion}

In this paper, we studied the performance of the Neyman-Pearson detection of a
stationary Gaussian process in noise, using a large wireless sensor network
(WSN). Our results are relevant for the design of sensor networks which are
constrained by limited signaling and communication overhead between the fusion
center and the sensor nodes.  We studied the case where each sensor compresses
its observation sequence using either a random iid linear precoder or an
orthogonal precoder.  In the random precoder case, we determined the error
exponent governing the asymptotic behaviour of the miss probability, when
$k,n\to\infty$ and $k/n\to c\in(0,1)$.  In the orthogonal precoder case, we
exhibit strategies (PCS and PFS) that achieve the best error exponent among all
orthogonal strategies. The PFS has moreover the attractive property of being
well suited for WSN with signaling overhead constraints.  In addition, we proved
that when the PFS is used, a low complexity test can be implemented at the FC as
an alternative to the Likelihood Ratio (Neyman-Pearson) test. Interestingly, the
proposed test performs as well as the LRT in terms of error exponents.

\section*{Acknowledgments}

We thank Jamal Najim for helpful comments and for bringing useful references to
our attention.

\appendices
\section{Proofs of main results}\label{sec:proofs-main-results}

Observe that we may set $\sigma^2=1$ without loss of generality, since it
amounts to divide $f$ by  $\sigma^2$ and the data ${\bs y}_i$ by $\sigma$.
Hence in the following proof sections, we assume $\sigma=1$. In particular the
LLR in~(\ref{eq:exprLLR}) for a precoding matrix $A_n$ with normalized
precoders ${\bs a}_i$, $i=1,\dots,k$ is given by
\begin{equation}
  \label{eq:exprLLRseq1}
2\L_{A_n}=\|{\bs z}\|^2 - \log \det(A_n^{T}\Gamma_nA_n+I_k)- 
{\bs z}^T(A_n^{T}\Gamma_nA_n+I_k){\bs z} \;.  
\end{equation}

\subsection{Proof of Theorem~\ref{the:iid}}
\label{sec:proofiid}

We assume without restriction that $\E((A_{11}^1)^2)=1$.
Due to Lemma~\ref{lem:Chen}, it is sufficient to prove that the normalized LLR
associated to strategy $\A$ converges in probability to the
rhs of equation~(\ref{eq:erriid}) under $H_0$. Expression~(\ref{eq:exprLLRseq1}) of
the LLR relies on the assumption that each precoder $\bs a_i$ has unit norm,
which is generally not the case for $A_n$ defined as in Theorem~\ref{the:iid}.
Since the false alarm and miss probabilities of this LRT
do not depend on the norms $\|\bs a_i\|$, $i=1,\dots,k$, it is equivalent to consider
precoders defined by the matrix
\begin{equation}
  \tilde A_n = A_n P_n^{-1/2}\ ,\label{eq:iid}
\end{equation}
where $P_n=\diag\left(\sum_{i=1}^n(A_{ij}^n)^2:{j=1,\dots,k}\right)$. 
With this definition, we may use expression~(\ref{eq:exprLLRseq1}) 
which is valid for normalized precoders.
In order to prove
Theorem~\ref{the:iid}, it is now sufficient to show that $-(1/n){\cal
  L}_{\tilde A_n}$ converges to the constant $K_{\mathrm{rnd}}$ defined in~(\ref{eq:erriid}).

The main issue lies in the asymptotic study of the two terms $\frac
1n\log\det(\tilde A_n^T\Gamma_n\tilde A_n+I_k)$ and $\frac 1n\bs
z^T(\tilde A_n^T\Gamma_n \tilde A_n+I_k)^{-1}\bs z$.  This can be done
by successively using the results of~\cite{hacMPRF}, \cite{sil95} and
\cite{hacAAP}. The crucial point is to characterize the limiting spectral
measure of matrix $\tilde A_n^T\Gamma_n\tilde A_n$. Define:
\begin{eqnarray*}
  R_n &=& \tilde A_n^T\Gamma_n\tilde A_n\ , \\
  S_n &=& \frac 1n A_n^T\Gamma_nA_n\ .
\end{eqnarray*}
First, we prove (see Lemma~\ref{lem:levybai} below) that the spectral measure of
$R_n$ is asymptotically close to the one of $S_n$ in a sense which is made
clear below.  Second, we apply the results of~\cite{yin86,sil95} along
with~\cite{GrSz} to determine the limiting spectral measure of $S_n$. Finally,
closed form expressions of the desired quantities follow from the results
of~\cite{hacAAP}.

Denote by $d$ the Lévy distance on the set of distribution functions. We recall
that $F_Q$ denotes the distribution function of the \emph{spectral measure} of
$Q$ (see Section~\ref{sec:iid}).  
\begin{lemma}
  \label{lem:levybai}
  As $n,k\to\infty$, ${d}(F_{R_n},F_{S_n})$ converges to zero in probability.
\end{lemma}
\begin{proof}
  The proof relies on Bai's formula (see \cite{bai93} and \cite{hacMPRF}) which
  provides the following bound on the Lévy distance:
\begin{equation}
  \label{eq:bai}
  {d}^4(F_{A^TA},F_{B^TB}) \leq \frac 2{n\,k}\,|A-B|^2(|A|^2+|B|^2)
\end{equation}
for any $n× k$ real matrices $A,B$, where $|A|=\sqrt{\Tr[A^TA]}$ denotes the
Frobenius norm of $A$.  We now use equation~(\ref{eq:bai}) with $A\gets
\Gamma_n^{1/2}\tilde A_n$ and $B\gets \frac 1{\sqrt
  n}\Gamma_n^{1/2}A_n$. Using~(\ref{eq:iid}) and introducing
$\Delta_n=(P_n^{-1/2}-\frac 1{\sqrt n}I_n)^2$, it is straightforward to show
that:
\begin{equation}
  \label{eq:baibis}
  {d}^4(F_{R_n},F_{S_n}) \leq \frac{2n}k\, \Tr [\Delta_nS_n]\ (\Tr [P_n^{-1}S_n] + \frac 1{n} \Tr S_n)\ .
\end{equation}
Note that $\frac 1n \Tr S_n \leq \frac {\rho(\Gamma_n)}{n^2} \Tr [A_n^TA_n]$,
where $\rho(\Gamma_n)$ denotes the spectral radius of $\Gamma_n$. Similarly,
$\Tr P_n^{-1}S_n = \Tr R_n\leq \rho(\Gamma_n)\, \frac 1 n \Tr \tilde
A_n^T\tilde A_n= \rho(\Gamma_n)\frac kn$.  Finally, $\Tr \Delta_nS_n\leq
\rho(\Gamma_n)\rho(\frac 1n A_n^TA_n)\Tr \Delta_n$. Putting all pieces together,
\begin{equation}
  {d}^4(F_{R_n},F_{S_n}) \leq  \kappa_n \Tr \Delta_n
\label{eq:boundD}
\end{equation}
where
$$
\kappa_n = \frac{2n}k\,\rho(\Gamma_n)^2\left(\frac kn+\frac 1{n^2} \Tr
  A_n^TA_n\right)\rho\left(\frac 1n A_n^TA_n\right)\ .
$$
From~\cite{GrSz}, $\rho(\Gamma_n)$ converges to $M_f = \sup(f)$.  By
the law of large numbers, $\frac 1{n^2} \Tr A_n^TA_n$ converges almost surely
(a.s.) to $c$.  From~\cite{yin88}, $\rho(\frac 1n A_n^TA_n)$ converges a.s. to
$(1+\sqrt{c})^2$.  Therefore, $\kappa_n$ converges a.s. to
$4M_f^2c(1+(1+\sqrt{c})^2)$.  In order to prove that ${d}(F_{R_n},F_{S_n})$
converges in probability to zero, it is sufficient, by
equation~(\ref{eq:boundD}), to prove that $\Tr\Delta_n$ converges in probability
to zero. We write $\Tr\Delta_n$ as follows:
$$
\Tr\Delta_n = \frac 1n \sum_{j=1}^k \xi_{n,j}
$$
where $\xi_{n,j} = \left(\left(\frac
    1n\sum_{i=1}^n(A_{ij}^n)^2\right)^{-1/2}-1\right)^2$.  Note that for a fixed
$n$, $\xi_{n,j}$ are iid for all $j$. Let $\epsilon>0$. By Markov inequality,
\begin{equation}
  \label{eq:cvprobaDelta}
  \bP\left(\Tr\Delta_n>\epsilon\right) \leq  \frac{\E\left(\sum_{j=1}^k \xi_{n,j}\right)}{n\epsilon}
  = \frac kn\, \frac{ \E\left(\xi_{n,1}\right)}{\epsilon}\ .
\end{equation}
As $\xi_{n,1}$ converges a.s. to zero, we conclude that
$\bP\left(\Tr\Delta_n>\epsilon\right)$ tends to zero.  This completes the proof
of Lemma~\ref{lem:levybai}.
\end{proof}

Thanks to Lemma~\ref{lem:levybai}, it is sufficient to study the asymptotic
behaviour of $F_{S_n}$. The latter is provided by~\cite{yin86,sil95}.  In order
to introduce this result, we need to recall some definitions.  For any
distribution function $F$, the Stieltjes transform $\mathrm{b}_F$ of $F$ is given
by:
$$
\mathrm{b}_F(z)=\int \frac{\rmd F(t)}{t-z}
$$
for each $z\in\C^+$, where $\C^+=\{z\in \C : \Im(z)>0\}$ with $\Im(z)$
denoting the imaginary part of~$z$.
Recall that, from the results
of~\cite{GrSz}, the spectral distribution function $F_{\Gamma_n}$ of $\Gamma_n$
converges weakly to $\Phi$ given by~(\ref{eq:cdf}).  By straightforward
application of the results of~\cite{yin86,sil95}, we obtain that, with
probability one, $F_{S_n}$ converges weakly to a deterministic measure $F$ whose
Stieltjes transform $\mathrm{b}=\mathrm{b}(z)$ is the unique solution in $\C^+$ of:
\begin{equation}
  \label{eq:jack}
  z=-\frac 1{\mathrm{b}} + \int\frac t{1+ct \mathrm{b}}\rm\rmd\Phi(t)\ ,
\end{equation}
for each $z\in\C^+$. The above result along with Lemma~\ref{lem:levybai} implies that
\begin{equation}
\label{eq:limFn}
\forall \epsilon>0,\ \bP({d}(F_{R_n},F)>\epsilon)\to 0\ .
\end{equation}
We are now in a position to study the limit of the LLR. We obtain immediately from~(\ref{eq:exprLLRseq1}):
\begin{equation}
  \label{eq:LLRdev}
    -\frac 1n\,\L_{\tilde A_n} = \frac k{2n}\left\{-\frac{\| \bs z\|^2}{k} + \beta_n + \gamma_n+ \delta_n\right\}
\end{equation}
where 
\begin{eqnarray*}
  \beta_n&=&\frac 1
  k \Tr \left(I_k+R_n\right)^{-1} = \int \frac 1{1+t}\,\rmd F_{R_n}(t)\\
  \gamma_n&=&\frac 1k\log\det\left(I_k+R_n\right) = \int \log\left(1+t\right)\rmd F_{R_n}(t) \\
  \delta_n&=&\frac 1k\bs z^T(I_k+R_n)^{-1}\bs z - \beta_n\ .
\end{eqnarray*}
Using~(\ref{eq:limFn}), $\beta_n$ and $\gamma_n$ respectively converge in probability
to the constants $\beta$ and $\gamma$ defined~by:
$$
\beta=\int \frac 1{1+t}\,\rmd F(t)\quad \textrm{ and }\quad \gamma=\int
\log\left(1+t\right)\rmd F(t)\ .
$$
Recalling that $\bs z\sim{\cal N}(0,I_k)$ under $H_0$, the term $\frac
1k \|\bs z\|^2$ in the rhs of~(\ref{eq:LLRdev}) converges
a.s. to one.  Since $\bs z$ is independent of $\tilde A_n$ and since the
spectral radius of $\left(R_n+I_k\right)^{-1}$ is bounded, it
is straightforward to show that $\delta_n\xrightarrow[]{a.s.}0$ (use for
instance Lemma~2.7 in~\cite{baiSil98}). 
Finally, $-(1/n)\,\L_{\tilde A_n}$ converges in probability to:
\begin{equation}
  \label{eq:LLRdevBis}
    K_\p({\cal A}) = \frac c{2}(-1 + \beta+ \gamma ) \ .
\end{equation}
Constant $\beta$ coincides with 
the Stieltjes transform of $F$ at point $-1$, that is
$\beta = \mathrm{b}(-1)$ where we defined for each $x<0$,
$\mathrm{b}(x)=\lim_{z\in\C^+\to x}\mathrm{b}(z)$. 
Constant $\beta$ is thus the unique solution
to~(\ref{eq:beta}).  A closed form expression for $\gamma$ can as well be obtained
using (for instance)~\cite{hacAAP}.  Using the fact that the limiting spectral
measure associated with $F$ has a bounded support, the dominated convergence
Theorem applies to the function $x\mapsto \int\log(x+t)dF(t)$. One easily obtains
after some algebra:
$$
\gamma = \int_{1}^\infty \left(\frac 1t - \mathrm{b}(-t)\right)\rmd t\ .
$$
Following~\cite{hacAAP}, we conclude that $\gamma = C(1)$ where $C$ is the
function defined for each $x>0$ by:
$$
C(x)=-1+x \mathrm{b}(-x)-\log(x\mathrm{b}(-x))+\frac 1c\int
\log(1+ct\mathrm{b}(-x))\rmd\Phi(t)\ .
$$
This statement can simply be proved by noting that $C'(t) = \frac 1t - \mathrm{b}(-t)$
(where $C'$ is the derivative of $C$) and $C(\infty)=0$.  Plugging the above
expression of $\gamma$ into~(\ref{eq:LLRdevBis}), we obtain the claimed
error exponent $K_{\mathrm{rnd}}$.

\subsection{Proof of Theorem~\ref{thm:optexp}}\label{sec:proof-theor-optexp}

We start with some useful definitions and technical preliminaries. 
Let $c\in(0,1)$. Denote  $m_f=\inf(f)$ $M_f=\sup(f)$ so that, by
definition of $\Phi$ in~(\ref{eq:cdf}), $\Phi(t) = 0$ for all $t< m_f$ and
$\Phi(t) = 1$ for all $t\geq M_f$. We  
define the set 
$$
\Lambda_c=\{\lambda\in(m_f,M_f]:\Phi(\lambda)\geq(1-c)\}\;.
$$ 
By assumptions {\bf A1} and {\bf A2}, $\Phi$ is continuously strictly
increasing from $[m_f,M_f]$ to $[0,1]$, we denote by $\Phi^{-1}$ its inverse
continuous function defined from $[0,1]$ to $[m_f,M_f]$. Hence
$\Lambda_c=[\Phi^{-1}(1-c), M_f]$. Moreover, using again {\bf A2}, we obtain
that $\1_{\Lambda_c}$ is almost surely continuous with respect to
$\Leb\circ f^{-1}$. By the uniform mapping theorem, this implies that, for any
sequence of probability measures $(\mu_n)$ weakly converging to $(2\pi)^{-1}\Leb\circ
f^{-1}$, we have, for all continuous function $g:\R\to\R$, as $n\to\infty$,  
\begin{equation}
  \label{eq:DeltaVSLambda}
\int_{\Lambda_c} g(\lambda) \rmd\mu_n(\lambda)\to
\frac1{2\pi}\int_{\Lambda_c} g(\lambda) \rmd \{\Leb\circ f^{-1}\}(\lambda)=  
\frac1{2\pi}\int_{\Delta_c}  g\circ f (\omega) \rmd \omega
\end{equation}
where the last equality follows from the definition of $\Delta_c$
in~(\ref{eq:Delta}) by setting $\lambda=f (\omega)$.

\subsubsection{The PCS case}

The outline of the proof is the following.

\begin{enumerate}[Step 1.]
\item Assume that
  \begin{equation}
    \label{eq:kcdefPCS}
k = \max\{i\in\{1,\dots,n\}\,:\, \Phi(\lambda_i^n)\geq 1-c\}\;,
  \end{equation}
where  $(\lambda_i^n)_{1\leq i\leq n}$ is given in
Definition~\ref{def:princ-comp-strat}. 
Then $k$  satisfies~(\ref{eq:growth-k-assump}) and strategy $\V$ 
has error exponent $K_{\mathrm{orth}}(c)$.\label{item:proof1}
\item Strategy $\V$ with any sequence $k$
  satisfying~(\ref{eq:growth-k-assump}) also has the error exponent
  $K_{\mathrm{orth}}(c)$. \label{item:proof2}
\item Under Condition~(\ref{eq:growth-k-assump}) strategy $\V$ is optimal
  among all orthogonal strategies, that is,~(\ref{eq:WF-err-exp-LB}) holds for
  any $\A$.\label{item:proof3}
\end{enumerate}

\noindent\textsl{Step \ref{item:proof1}}.
Let $\mu_n = \frac1n \sum_{i=1}^n \delta_{\lambda_i^n}$ denote the empirical
spectral measure of $\Gamma_n$ defined in~(\ref{eq:gammafctf}). Szegö's Theorem states that $\mu_n$ converges
weakly to $\frac1{2\pi}\Leb\circ f^{-1}$ (\cite{GrSz},
p.64). Applying~(\ref{eq:DeltaVSLambda}) and then Lemma \ref{lem:DeltaC} then gives 
\begin{equation*}
\frac1n \sum_{i=1}^n \1_{\Lambda_c}(\lambda_i^n)= 
\frac{1}{2\pi}\int_{\Delta_c}\rmd\omega = c\;.
\end{equation*}
That is, $k$ defined by~(\ref{eq:kcdefPCS}) satisfies~(\ref{eq:growth-k-assump}).
Recall that here $V_n$ is given in Definition~\ref{def:princ-comp-strat} with
$k$ given by in~(\ref{eq:kcdefPCS}). The empirical spectral measure of
$V_n^{T}\Gamma_nV_n+I_k$ is thus
given by $\frac1n \sum_{i=1}^n
\delta_{1+\lambda_i^n}\1_{\Lambda_c}(\lambda_i^n)$. Hence we have as above that
\begin{align}\label{eq:LimiteLL2}
&  \lim_{n\to\infty}\frac1n\log\det(V_n^{T}\Gamma_nV_n+I_k) =
  \frac1{2\pi}\int_{\Delta_c}\log(1 + f(\omega))\rmd\omega\;,\\
\label{eq:LimiteLL4}
&  \lim_{n\to\infty}\Tr\left[(V_n^{T}\Gamma_nV_n+ I_k)^{-1}\right]= 
  \frac1{2\pi}\int_{\Delta_c}\frac1{1 + f(\omega)}\rmd\omega\;.
\end{align}
The spectral radius $\rho[(V_n^{T}\Gamma_nV_n+I_k)^{-1}]$ is bounded
by $1/(1+ \Phi^{-1}(1-c))$. Using eqs~(\ref{eq:exprLLRseq1}),
(\ref{eq:LimiteLL2})--(\ref{eq:LimiteLL4}), Lemma \ref{lem:QuadGauss}
and~(\ref{eq:WF-err-exp}), we obtain that, under $H_0$, $-\frac1n\L_{V_n}\cp K_{\mathrm{orth}}(c)$.  
As a consequence of Lemma \ref{lem:Chen}, we obtain the assertion of Step~1.

\noindent\textsl{Step \ref{item:proof2}}. Observe that the error exponent associated to
a strategy $\V$ is increasing with $k$. Now let $k$ be a sequence
satisfying~(\ref{eq:growth-k-assump}). For any $c'$ and $c''$  such that
$c'<c<c''$,  define $k'$ and $k''$  by~(\ref{eq:kcdefPCS}) with $c$ replaced by
$c'$ and $c''$ respectively. Then, as seen in Step~\ref{item:proof1}, $k'$ and
$k''$ also satisfy~(\ref{eq:growth-k-assump}) with $c$ replaced by $c'$ and $c''$ respectively.
Thus, eventually, $k'\leq k\leq k''$, and, applying Step~\ref{item:proof1}, the
error exponent of $\V$ belongs to
$[K_{\mathrm{orth}}(c'),K_{\mathrm{orth}}(c'')]$. This, with the continuity of
$K_{\mathrm{orth}}$, yields the assertion of Step \ref{item:proof2}.  
      
\noindent\textsl{Step \ref{item:proof3}}. Assume $\A = (A_n)$ now denotes
\emph{any} orthogonal, that is $A_n$ is a $n× k$ orthogonal matrix for all
$n$, where $k$ satisfies~(\ref{eq:growth-k-assump}). Let us prove that the
bound~(\ref{eq:WF-err-exp-LB}) holds.  By Lemma~\ref{lem:Chen}, for all real
$t$, we have
\begin{equation}
  \label{eq:UBerrexp}
\liminf_{n\to\infty}\bP_0\left[-\frac1n\L_{A_n} >  t\right] = 0 \Rightarrow \overline K_{\p}(\A)\leq t\;.  
\end{equation}
Let $t>\limsup_{n\to\infty}\E_0\left[-\L_{A_n}/n\right]$.
Using Markov inequality, we have for $n$ large enough,
$$
\bP_0\left[-\frac1n\L_{A_n} > t\right] \leq
\frac{\Var_0\left[\frac1n\L_{A_n}\right]}{(t+\E_0\left[\frac1n\L_{A_n}\right])^2}\;.
$$
Using~(\ref{eq:exprLLRseq1}), we have 
$$
\Var_0\left[\frac1n \L_{A_n}\right]
=\frac14 \Var_0\left[{\bs z}^T \left\{I_k-(I_k+A_n^T\Gamma_nA_n)^{-1}\right\}{\bs z}\right]
\stackrel{n\to+\infty}{\to}0\;,
$$
where the convergence follows from Lemma \ref{lem:QuadGauss} by noticing that
$I_k-(I_k+A_n^T\Gamma_nA_n)$ has  eigenvalues in
$[0,1]$ and, under $H_0$ (recall that we set $\sigma^2=1$), ${\bs z}\sim \mathcal{N}(0,I_k)$.
The last two displays show that $\bP_0\left[-\frac1n\L_{A_n} > t\right]\to 0$ as
$n\to\infty$ for all
$t>\limsup_{n\to\infty}\E_0\left[-\L_{A_n}/n\right]$. With~(\ref{eq:UBerrexp}),
we
get that
$$
\overline K_{\p}(\A)\leq \limsup_{n\to\infty}\E_0\left[-\L_{A_n}/n\right] \;.
$$
To conclude the proof, it thus only remains to show that
$\limsup_{n\to\infty}\E_0\left[-\L_{A_n}/n\right]\leq K_{\mathrm{orth}}(c)$. 
We have, by~(\ref{eq:exprLLRseq1}),
$$
\E_0\left[-\L_{A_n}\right]=-k
+\log\det(A_n^T\Gamma_nA_n+I_k)+\Tr\left( (A_n^T\Gamma_nA_n+I_k)^{-1} \right)\;.
$$
Since $x\mapsto \log x + 1/x$ is nondecreasing on
$[1,+\infty[$, Lemma \ref{lem:HornJohnson} thus implies that
$\E_0\left[-\L_{A_n}\right]\leq\E_0\left[-\L_{V_n}\right]$. We proved in
Step~\ref{item:proof1}
that $\E_0\left[-\L_{V_n}/n\right]\to K_{\mathrm{orth}}(c)$. Hence the proof is
achieved. 

\subsubsection{The PFS case}

We now prove that the PFS strategy also achieves the error exponent
$K_{\mathrm{orth}}(c)$ under the condition~(\ref{eq:growth-k-assump}).
Using the same argument as in Step~\ref{item:proof2} of the PCS case, we can in
fact take $k$ as defined by
  \begin{equation}
    \label{eq:kcdefPFS}
k = \max\{i\in\{1,\dots,n\}\,:\, \Phi\circ f(2\pi j_{i-1}^n)\geq 1-c\}\;,
  \end{equation}
where $(j_i^n)_{0\leq i<n}$ is given in Definition~\ref{def:princ-freq-strat},
which we assume in the following.  

It is known (\cite{Gutierrez:Crespo:2008}, Lemma 4.6) that $\Gamma_n$ defined
in~(\ref{eq:gammafctf}) is
asymptotically equivalent to $F_n^TD_nF_n$ where $D_n$ denotes the $n× n$
diagonal matrix with entries $f(2\pi k/n)$, $k=0,\dots,n-1$. As in
\cite{Gray06}, we denote asymptotic equivalence between matrices $A_n$ and $B_n$
by $A_n\sim B_n$. Asymptotic equivalence is preserved by elementary matrix
operations (\cite{Gray06}, Proposition 2.1). Hence, $F_n$ being unitary,
$$D_n\sim F_n\Gamma_nF_n^T\;.$$

Also, from the definition of $W_n$,
$$W_n = F_n^TS_n\;,$$
where $S_n$ is a $n× k$ selection matrix the columns of which belong to the
canonical basis.  Hence, $S_n$ being unitary,
\begin{equation}\label{eq:ToepCircEq}
S_n^TD_nS_n\sim W_n^T\Gamma_nW_n\;.
\end{equation}

Eq (\ref{eq:ToepCircEq}) implies
\begin{eqnarray*}
  \lim_{n\to\infty}\frac1n\log\det(W_n^T\Gamma_nW_n+I_k) &=&
  \lim_{n\to\infty}\frac1n\sum_{\PF^n(f;c)}\log(1+f(2\pi k/n))\\
  & = & \frac1{2\pi}\int_{\Delta_c}\log(1+f(\omega))\rmd\omega\;,
\end{eqnarray*}
And
\begin{eqnarray*}
  \lim_{n\to\infty}\Tr\left[(W_n^T\Gamma_nW_n+I_k)^{-1}\right] &=& 
  \lim_{n\to\infty}\frac1n\sum_{\PF^n(f;c)}\frac{1}{1+f(2\pi k/n)}\\
  & = &
  \frac1{2\pi}\int_{\Delta_c}\frac1{1 + f(\omega)}\rmd\omega\;.
\end{eqnarray*}

We then conclude as in Step~\ref{item:proof1} of the PCS case.

\subsection{Proof of
  Theorem~\ref{thm:simpleStat}}\label{sec:proof-theor-statSimple}

Again we can take $k$ defined by~(\ref{eq:kcdefPCS}) without loss of
generality. 

Let $D_n$ denote the $n× n$ diagonal matrix with entries $D_n(\ell,\ell)=
\tilde{f}(2\pi \ell/n)$, $\ell=0,\dots,n-1$, where we defined the function
$$
\tilde{f}(\omega)=
\left(\frac1{1+\sigma^2}-\frac1{1+\sigma^2+f(\omega)}\right)\1(\Phi\circ
f(\omega)\geq 1-c)\;.
$$
Then by Definition~\ref{def:princ-freq-strat} and~(\ref{eq:simpleStat}), we have
$$
\mathcal{T}_n\stackrel{d}{=}{\sf u}_n^T M_n {\sf u}_n \;,
$$
where $M_n=F_n^T D_n F_n$ and ${\sf u}_n$ is a $n$-sample of a centered
stationary Gaussian process with spectral density
$g_1(\omega)=1+\sigma^2+f(\omega)$ under $H_1$ and $g_0(\omega)=1+\sigma^2$
under $H_0$. We shall apply \cite[Propostion~2]{bercu:gamboa:Rouault:1997} which
provides a large deviation principle (LDP) for quadratic forms of stationary
Gaussian processes. Recall that we denote by $T_n(g)$ the $n× n$ covariance
matrix associated to the spectral density $g$ (see~(\ref{eq:gammafctf})). Let
$S_n=\mathrm{Sp}(T_n(g)^{1/2}M_nT_n(g)^{1/2})$ denote the set of eigenvalues of
$T_n(g)^{1/2}M_nT_n(g)^{1/2}$. Since $M_n$ is non-negative, to apply this
result, we successively show that for $g=g_0$ or $g=g_1$,
\begin{enumerate}[(i)]
\item $\bar{a}_n=\max(S_n)$ is bounded above by $M_{\tilde{f}}M_{g}$,
\item the following weak convergence holds $n^{-1}\sum_{\lambda\in
    S_n}\delta_\lambda \Rightarrow \frac1{2\pi}\Leb\circ[\tilde{f}g]^{-1}$,
\item $\bar{a}_n\to M_{\tilde{f}g}$ as $n\to\infty$. 
\end{enumerate}
Observe that the eigenvalues of $D_n$ are given by $\tilde{f}(2\pi \ell/n)$,
with $\ell=0,\dots,n-1$, and those of $T_n(g)$ are bounded by $M_g$. Hence we
have (i). Assertion (ii) is a consequence of Lemma~5 in
\cite{Gutierrez:Crespo:2008} and Theorem~2.1 in\cite{Gray06}. By (i) and (ii),
we have
$$
\limsup \bar{a}_n\leq M_{\tilde{f}}M_{g}\quad\text{and}\quad
M_{\tilde{f}g}\leq \liminf \bar{a}_n \;.
$$
Thus Assertion (iii) follows by observing that $f$, $g_0$ and $g_1$ achieve
their maxima at the same points, thus $M_{\tilde{f}g}= M_{\tilde{f}}M_{g}$ for
$g=g_0$ or $g_1$. Since Assertions (i)--(iii) hold, Propostion~3 and Corollary~2
in \cite{bercu:gamboa:Rouault:1997} give that for $i=0,1$, under $H_i$,
$n^{-1}\mathcal{T}_n$ satisfies a LDP with good rate function
\begin{equation}
  \label{eq:rateFunc}
  I_i(x)=\sup_{y\in\R}\left(yx+\frac1{4\pi}\int_{-\pi}^\pi \log(1-2 y[\tilde{f}g_i](\omega))\,\rmd\omega\right) \;. 
\end{equation}
with $g=g_0$ or $g=g_1$ under $H_0$ or $H_1$, respectively. As in
\cite{bercu:gamboa:Rouault:1997}, we assume for convenience that
$\log(x)=-\infty$ when $x\leq 0$.

Assertion (ii) above also implies that
$n^{-1}\mathcal{T}_n\stackrel{P}{\to}\frac{1}{2\pi}\int_{-\pi}^\pi[\tilde{f}g](\omega)\rmd\omega$
with the same convention for $g$. Hence the sequence $(\eta_n(\p))$ in
Theorem~\ref{thm:simpleStat} satisfies $n^{-1}\eta_n(\p)\to
x_0:=\frac{1}{2\pi}\int_{-\pi}^\pi[\tilde{f}g_0](\omega)\rmd\omega$. Thus the
LDP under $H_1$ gives
$$
\limsup_{n\to\infty} \frac1n\log\bP_1(\mathcal{T}_n\leq\eta_n(\p)) \leq
\inf_{c>x_0} \limsup_{n\to\infty}\frac1n\log\bP_1(n^{-1}\mathcal{T}_n\leq c)
\leq\inf_{c>x_0}\sup_{x\leq c} -I_1(x) \;,
$$
and
$$
\liminf_{n\to\infty} \frac1n\log\bP_1(\mathcal{T}_n\leq\eta_n(\p)) \geq
\sup_{c<x_0}\liminf_{n\to\infty} \frac1n\log\bP_1(n^{-1}\mathcal{T}_n< c)
\geq\sup_{c<x_0}\sup_{x<c} -I_1(x) \;,
$$
By Lemma~\ref{lem:contRateFunc}, we conclude that
$$
\lim_{n\to\infty} -\frac1n\log\bP_1(\mathcal{T}_n\leq\eta_n(\p))= \inf_{x\leq
  x_0} I_1(x)\;.
$$
To conclude the proof, it only remains to show that $I_1(x_0)=K_\p(\W^c)$ and
$I_1(x)$ is nonincreasing on $(-\infty, x_0]$.  By definition of $\tilde{f}$
and $I_1$, we have $I_1(x)=\sup_{y}F_x(y)$, where
$$
F_x(y)=yx+\frac1{4\pi}\int_{\Delta_c} \log(1-2 y[g_1/g_0-1](\omega))\,\rmd\omega\;.
$$
Using the definition of $x_0$, we further have $yx_0=\int_{\Delta_c}
y[1-g_0/g_1](\omega)\rmd\omega$ and hence $F_{x_0}(y)=\frac1{4\pi}\int_{\Delta_c}
f_\omega(2y)\,\rmd\omega$ with
$$
f_\omega(y)= y [1-g_0/g_1](\omega)+\log(1- y[g_1/g_0-1](\omega))\;.
$$
For any $\omega$, it is straightforward to show that $f_\omega(y)$ is maximized at
$y=-1$ at which it takes value
$f_\omega(-1)=2D\left(\N(0,g_0(\omega))\,||\,\N(0,g_1(\omega)\right)$. Since the
maximizing $f_\omega(y)$ does not depend on $\omega$ we obtain
\begin{align*}
I_1(x_0)&=\sup_{y\in\R}\frac1{4\pi}\int_{\Delta_c} f_\omega(2y)\,\rmd\omega\\
&=\frac1{4\pi}\int_{\Delta_c}\sup_{y\in\R} f_\omega(2y)\,\rmd\omega
=K_\p(\W^c)\,.
\end{align*}
We now consider $x\leq x_0$. By differentiating $F_x$, the $y$ maximizing
$F_x(y)$ satisfies
$$
x=\int_{\Delta_c}\frac{[g_1/g_0-1](\omega)}{1-2y[g_1/g_0-1](\omega)}\,\rmd\omega\;.
$$
Note that $g_1/g_0-1$ is non-negative and has a positive integral on $\Delta_c$
hence the right-hand side of the previous display has a strictly positive
derivative w.r.t. $y$. It follows that $y(x)$, defined as the $y$ maximizing
$F_x(y)$, is strictly increasing with $x$. On the other hand, we know from above
that $y(x_0)=-1/2$. Thus, for all $x\leq x_0$, we have
$I_1(x)=\sup_{y\leq-1/2}F_x(y)$. Now observe that for all $x'\leq x$ and
all $y\leq0$ we have $F_{x'}(y)-F_x(y)=(x'-x)y\geq0$. It follows that $I_1$ is
nonincreasing on $(-\infty,x_0]$, which achieves the proof.

\section{Technical lemmas}

\begin{lemma}\label{lem:QuadGauss}
  Assume that for each $n>0$, $x_n\sim\N(0,\Sigma_n)$ where $\Sigma_n$ has
  bounded spectral radius $\rho(\Sigma_n)$; and assume $Q_n$ is a family of
  quadratic forms with bounded spectral radius $\rho(Q_n)$. Then,
  $$\Var[\frac1n x_n^TQ_nx_n] \stackrel{n\to+\infty}{\to}0\;.$$
  If, moreover,
  $$
  \lim_{n\to\infty}\frac1n \Tr(Q_n\Sigma_n)\to c\;,
  $$
  Then $\frac1n(x_n^TQ_nx_n)$ converges in the $L_2$ sense towards $c$.
\end{lemma}
\begin{IEEEproof}
  One has $\E[\x_n^TQ_n\x_n] = \Tr[Q_n\Sigma_n]$. Let us estimate
  $\Var[\x_n^TQ_n\x_n]$.
  $$\x_n^TQ_n\x_n = y_n^T\Delta_ny_n$$
  with $y_n$ a standard centered gaussian vector and $\Delta_n$ diagonal and
  congruent to $\Sigma_n^{\frac12}Q_n\Sigma_n^{\frac12}$.
  $$\Var[\x_n^TQ_n\x_n] = 2\Tr[\Delta_n^2] \leq 2n\rho(\Delta_n^2)\leq
  2n\rho(\Sigma_n)^2\rho(Q_n)^2\leq C\cdot n$$
  where $C$ is a constant.
  Thus we have, as sought,
  $$\Var[\frac1n\x_n^TQ_n\x_n]\to_{n\to\infty}0\;.$$
\end{IEEEproof}

\begin{lemma}[\cite{HoJo}, p. 189]\label{lem:HornJohnson}
  Let $Q$ be a symmetric $n× n$ matrix, and $V$ be a $r$-dimensional subspace of
  $\R^n$.  Denote by $Q_V$ the restriction of $Q$ to $V$, $\lambda_i$,
  $i\in\{1,\dots,n\}$, the eigenvalues of $Q$ in increasing order and $\mu_j$,
  $j\in\{1,\dots,r\}$, the eigenvalues of $Q_V$ in increasing order.  Then, for
  all $i=1,\dots r$, we have $\lambda_i\leq\mu_i \leq\lambda_{n+i-r}$.
\end{lemma}

\begin{lemma}\label{lem:DeltaC}
  Under {\bf A1-A2}, we have for any $c\in[0,1]$, $\Leb(\Delta_c)=2\pi c$, where
  $\Delta_c$ is defined by~(\ref{eq:Delta}).
\end{lemma}
\begin{IEEEproof}
  We have $\Delta_c= (\Phi\circ f)^{-1}([1-c,\infty))\cap(-\pi,\pi)$, where
  $(\Phi\circ f)^{-1}$ denotes the inverse image under $\Phi\circ f$.  Observe
  that $(\Phi\circ f)^{-1}=f^{-1}\circ \Phi^{-1}$. Moreover as we have seen in
  the preamble of Appendix~\ref{sec:proof-theor-optexp}, $\Phi$ is continuously
  and strictly increasing from $[m_f,M_f]$ to $[0,1]$ and constant on $[M_f,\infty)$, hence 
  $\Phi^{-1}([1-c,\infty))=[\Phi^{-1}(1-c),\infty)$, where $\Phi^{-1}$ here
  denotes the inverse function from $[0,1]$ to $[0,M_f]$. Hence
  $\Delta_c=f^{-1}([\Phi^{-1}(1-c),\infty)$. Now since $\Phi$ is the
  distribution function of the probability measure $(2\pi)^{-1}\Leb\circ f^{-1}$
  and using again that it is continuously and strictly increasing  from
  $[0,M_f]$ to $[0,1]$, we get that $(2\pi)^{-1}\Leb(\Delta_c)=1-(1-c)=c$, which
  concludes the proof.
\end{IEEEproof}

\begin{lemma}\label{lem:contRateFunc}
  Let $I(x)$ be defined for $x\in\R$ by~(\ref{eq:rateFunc}) with values in
  $\R\cup\{\infty\}$ for some non-negative bounded function
  $h=[\tilde{f}g]$. Then $I(x)$ is finite and continuous for $x>0$.
\end{lemma}
\begin{IEEEproof}
  Let $J_x(y)=yx+\frac1{4\pi}\int_{-\pi}^\pi
  \log(1-2y[\tilde{f}g](\omega))\,\rmd\omega$ so that $I(x)=\sup_yJ_x(y)$. Let
  $M_h$ denote the essential sup of $h$. Then $J_x(y)=-\infty$ for all
  $y>1/(2M_h)$.  Let $\epsilon>0$. Note that $J_x(0)=0$ and for all
  $x\geq\epsilon$ and $y\leq0$, $J_x(y)\leq yx+\log(1-2 y M_h)/2\to-\infty$ as
  $y\to-\infty$. Thus there exists $y_\epsilon$ only depending on $\epsilon$
  such that $J_x(y)\leq 0$ for all $x\geq\epsilon$ and $y\leq y_\epsilon$. From
  these facts, it follows that for all $x\geq\epsilon$,
  $I(x)=\sup_{y\in[y_\epsilon,1/(2M_h)]}|J_x(y)|$. Finally we observe that for
  all $x,x'\geq\epsilon$,
  $\sup_{y\in[y_\epsilon,1/(2M_h)]}|J_x(y)-J_{x'}(y)|\leq
  (-y_\epsilon\lor1/(2M_h))\,|x-x'|$ which now yields the result.
\end{IEEEproof}

{\small 
\bibliographystyle{IEEEbib}
\bibliography{Biblio}}
\end{document}